\documentclass[11pt]{article}

\usepackage{amsmath}
\usepackage{amsfonts}
\usepackage{amssymb}
\usepackage{amscd}
\usepackage{latexsym}
\usepackage{mathrsfs}
\usepackage{amsthm}
\usepackage{comment}
\usepackage{empheq}
\usepackage{tikz}
\usepackage{pst-poly,multido}
\usepackage{subfigure}

\usetikzlibrary{graphs}
\usetikzlibrary{decorations.pathreplacing,matrix,positioning}
\usetikzlibrary{calc}
\usetikzlibrary{intersections}

\newcommand\ie{{\em i.e.}}

\def\A{{A}}

\def\N{\mathbb{N}}
\def\T{\mathbb{T}}

\def\Z{\mathbb{Z}}
\def\C{\mathbb{C}}
\def\D{{\mathcal D}}

\def\F{\mathscr F}
\def\f{u}
\def\H{\mathcal H}
\def\HH{\mathscr H}
\def\K{\mathcal K}
\def\e{{\mathrm e}}

\def\R{\mathbb{R}}

\def\X{\mathcal X}

\def\LL{\mathfrak L}
\def\U{\mathscr U}

\def\KK{\mathfrak K}

\def\NNN{\mathfrak N}

\def\Id{\mathbb I}
\def\I{\mathscr I}

\def\B{\mathcal B}

\def\({\left(}
\def\[{\left[}
\def\){\right)}
\def\]{\right]}
\def\lp{\left\lVert}
\def\rp{\right\rVert}

\def\GG{\mathscr{G}}

\def\<{\langle}
\def\>{\rangle}

\def\Op{\mathfrak{Op}}

\def\XX{\mathfrak{X}}

\def\xx{\mathfrak{x}}

\def\ee{\mathfrak{e}}

\def\J{\mathscr{J}}

\def\d{\mathrm d}

\DeclareMathOperator{\Aut}{Aut}
\DeclareMathOperator{\Ran}{Ran}

\newtheorem{Theorem}{Theorem}[section]

\newtheorem{Lemma}[Theorem]{Lemma}

\newtheorem{Proposition}[Theorem]{Proposition}
\newtheorem{Definition}[Theorem]{Definition}
\newtheorem{Example}[Theorem]{Example}

\begin{document}

\title{Spectral and scattering theory for topological crystals \\
perturbed by infinitely many new edges}

\author{S. Richard\footnote{Supported by the grant\emph{Topological invariants
through scattering theory and noncommutative geometry} from Nagoya University,
and by JSPS Grant-in-Aid for scientific research C no 18K03328 \& 21K03292, and on
leave of absence from Univ.~Lyon, Universit\'e Claude Bernard Lyon 1, CNRS UMR 5208,
Institut Camille Jordan, 43 blvd.~du 11 novembre 1918, F-69622 Villeurbanne cedex,
France.}, N. Tsuzu}

\date{\small}
\maketitle \vspace{-18mm}

\begin{quote}
\emph{
\begin{itemize}
\item[] Graduate school of mathematics, Nagoya University,
Chikusa-ku,  Nagoya 464-8602, Japan
\item[] \emph{E-mail:} richard@math.nagoya-u.ac.jp, 
m20030u@math.nagoya-u.ac.jp
\end{itemize}
}
\end{quote}

\begin{abstract}
In this paper we investigate the spectral and scattering theory for operators 
acting on topological crystals and on their perturbations. 
A special attention is paid to perturbations obtained by the addition of an infinite number of  edges, and\,/\,or  by the removal of a finite number of them,
but perturbations of the underlying measures and perturbations by the addition of a multiplication operator are also considered. The description of the nature of the 
spectrum of the resulting operators, and the existence and completeness of the wave operators are standard outcomes for these investigations.
\end{abstract}

\smallskip


\section{Introduction}\label{sec_intro}
\setcounter{equation}{0}

Consider a topological crystal, namely a perfect periodic discrete structure of arbitrary dimension, and let $H_0$ be a discrete Schr\"{o}dinger type operator acting on it.
Properties of such systems are well known, and the band structure of the spectrum
of $H_0$ has been studied for decades, see for example \cite[Sec.~XIII.16]{RS4}
for an introduction to the subject. Perturbations of such systems have also been 
extensively studied, often for a restricted family of graphs but also in the general
framework of topological crystals, see the list of references mentioned below.
Most of the time, the perturbations considered were either modifications of the weights
supported by the vertices or by the edges of the graph, or a perturbation
due to the addition of a potential decaying at infinity. Quite rarely, modifications 
of the graph itself were considered, and almost always structural perturbations
were confined in a bounded domain.

An extreme situation which has not been considered so far is about the addition of an infinite number of edges. If we think about a topological crystal as a perfectly ordered structure, with each vertex linked regularly to a very small number of neighbors, the addition of an infinite number of edges means a possible interaction
between vertices which are very far away from each others. Such systems can now 
describe long distance interactions, and can be used for modeling a much larger
family of weakly interacting physical systems. Clearly, the addition of an infinite number of edges can be performed only if suitable weights on them are imposed. The sum of the weights can not grow too much locally or at infinity. Similarly, it is possible to remove a few edges from the initial perfect lattice, but removing an infinite number of them is not possible: if we assume for a second that all vertices and initial edges have a weight $m=1$, then removing an infinite number of edges would produce an operator $H$ no more comparable with $H_0$, and perturbation theory would not apply anymore. 

Before describing more precisely the content of this paper, let us propose two examples on $\Z^d$ which provide an idea about typical conditions appearing when
an infinite number of edges are added, either connected to one vertex or to all vertices.
The Euclidean norm in $\Z^d$ is simply denoted by $|\!\cdot\!|$.

\begin{Example}[$0\in\Z^{d}$ connects to all other vertices]\label{ex_0_connect}
We consider the lattice $\Z^{d}$ and add infinitely many edges connecting $0$ to all other vertices, as shown in Figure \ref{pic_0_connect} for $d=1$. 
The set of added edges is denoted by $F$. 
We also fix $m(x)=1$ for any $x\in\Z^{d}$ and  
$m(\e)=1$ for all initial edges $\e$. For any $\e \in  F$, with 
endpoints $0$ and $y$, we assume that  $m(\e) \leq C | y |^{\alpha}$
for some $\alpha<-d-2$ and some constant $C$ independent of $y$. 

\begin{figure}[h]
\begin{center}
\tikzset{node/.style={circle,fill=black,inner sep=1pt}}
\begin{tikzpicture}
\coordinate (s) at (0, 0);
\coordinate[node, label=below:{$-4$}] (-4) at (-4, 0);
\coordinate[node, label=below:{$-3$}] (-3) at (-3, 0);
\coordinate[node, label=below:{$-2$}] (-2) at (-2, 0);
\coordinate[node, label=below:{$-1$}] (-1) at (-1, 0);
\coordinate[node, label=below:{$0$}] (0) at (0, 0);
\coordinate[node, label=below:{$1$}] (1) at (1, 0);
\coordinate[node, label=below:{$2$}] (2) at (2, 0);
\coordinate[node, label=below:{$3$}] (3) at (3, 0);
\coordinate[node, label=below:{$4$}] (4) at (4, 0);
\draw (-4.5, 0)--(4.5, 0);

\coordinate (s) at (0, 0); 
\draw ([shift={($(0.5, 0)+(s)$)}] 0:0.5) arc [x radius=0.5, y radius=0.2, start angle = 0, end angle=180];
\draw ([shift={($(-0.5, 0)+(s)$)}]0:0.5) arc [x radius=0.5, y radius=0.2, start angle = 0, end angle=180];
\draw ([shift={($(1, 0)+(s)$)}] 0:1) arc [x radius=1, y radius=0.4, start angle = 0, end angle=180];
\draw ([shift={($(-1, 0)+(s)$)}] 0:1) arc [x radius=1, y radius=0.4, start angle = 0, end angle=180];
\draw ([shift={($(1.5, 0)+(s)$)}] 0:1.5) arc [x radius=1.5, y radius=0.6, start angle = 0, end angle=180];
\draw ([shift={($(-1.5, 0)+(s)$)}] 0:1.5) arc [x radius=1.5, y radius=0.6, start angle = 0, end angle=180];
\draw ([shift={($(2, 0)+(s)$)}] 0:2) arc [x radius=2, y radius=0.8, start angle = 0, end angle=180];
\draw ([shift={($(-2, 0)+(s)$)}] 0:2) arc [x radius=2, y radius=0.8, start angle = 0, end angle=180];
\draw ([shift={($(2.5, 0)+(s)$)}] 0:2.5) arc [x radius=2.5, y radius=1, start angle = 0, end angle=180];
\draw ([shift={($(-2.5, 0)+(s)$)}] 0:2.5) arc [x radius=2.5, y radius=1, start angle = 0, end angle=180];
\end{tikzpicture}
\end{center}
\caption{$0\in \Z$ is connected to all other vertices}
\label{pic_0_connect}
\end{figure}
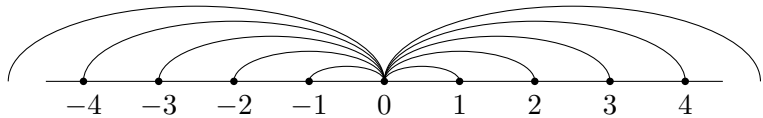
\end{Example}

\begin{Example}[All vertices of $\Z^{d}$ connected to all other vertices]\label{ex_all_connect}
We consider the lattice $\Z^{d}$, and add infinitely many edges connecting all vertices to each others.  
The set of added edges is denoted by $F$. 
We also fix $m(x)=1$ for any $x\in\Z^{d}$ and
$m(\e)=1$ for all initial edges $\e$. For any $\e \in  F$, with 
endpoints $x$ and $y$, we assume that  $m(\e) \leq C (1+|x|)^{\alpha}(1+|y|)^{\alpha}$
for some $\alpha<-d-2$ and some constant $C$ independent of $\e$.
\begin{figure}[h]
\begin{center}
\tikzset{node/.style={circle,fill=black,inner sep=1pt}}
\begin{tikzpicture}
\coordinate (s) at (0, 0);
\coordinate[node, label=below:{$-4$}] (-4) at (-4, 0);
\coordinate[node, label=below:{$-3$}] (-3) at (-3, 0);
\coordinate[node, label=below:{$-2$}] (-2) at (-2, 0);
\coordinate[node, label=below:{$-1$}] (-1) at (-1, 0);
\coordinate[node, label=below:{$0$}] (0) at (0, 0);
\coordinate[node, label=below:{$1$}] (1) at (1, 0);
\coordinate[node, label=below:{$2$}] (2) at (2, 0);
\coordinate[node, label=below:{$3$}] (3) at (3, 0);
\coordinate[node, label=below:{$4$}] (4) at (4, 0);
\draw (-4.5, 0)--(4.5, 0);

\coordinate (s) at (0, 0); 
\draw ([shift={($(0.5, 0)+(s)$)}] 0:0.5) arc [x radius=0.5, y radius=0.2, start angle = 0, end angle=180];
\draw ([shift={($(-0.5, 0)+(s)$)}]0:0.5) arc [x radius=0.5, y radius=0.2, start angle = 0, end angle=180];
\draw ([shift={($(1, 0)+(s)$)}] 0:1) arc [x radius=1, y radius=0.4, start angle = 0, end angle=180];
\draw ([shift={($(-1, 0)+(s)$)}] 0:1) arc [x radius=1, y radius=0.4, start angle = 0, end angle=180];
\draw ([shift={($(1.5, 0)+(s)$)}] 0:1.5) arc [x radius=1.5, y radius=0.6, start angle = 0, end angle=180];
\draw ([shift={($(-1.5, 0)+(s)$)}] 0:1.5) arc [x radius=1.5, y radius=0.6, start angle = 0, end angle=180];
\draw ([shift={($(2, 0)+(s)$)}] 0:2) arc [x radius=2, y radius=0.8, start angle = 0, end angle=180];
\draw ([shift={($(-2, 0)+(s)$)}] 0:2) arc [x radius=2, y radius=0.8, start angle = 0, end angle=180];
\draw ([shift={($(2.5, 0)+(s)$)}] 0:2.5) arc [x radius=2.5, y radius=1, start angle = 0, end angle=180];
\draw ([shift={($(-2.5, 0)+(s)$)}] 0:2.5) arc [x radius=2.5, y radius=1, start angle = 0, end angle=180];

\coordinate (s) at (1, 0); 
\draw ([shift={($(0.5, 0)+(s)$)}] 0:0.5) arc [x radius=0.5, y radius=0.2, start angle = 0, end angle=180];
\draw ([shift={($(-0.5, 0)+(s)$)}]0:0.5) arc [x radius=0.5, y radius=0.2, start angle = 0, end angle=180];
\draw ([shift={($(1, 0)+(s)$)}] 0:1) arc [x radius=1, y radius=0.4, start angle = 0, end angle=180];
\draw ([shift={($(-1, 0)+(s)$)}] 0:1) arc [x radius=1, y radius=0.4, start angle = 0, end angle=180];
\draw ([shift={($(1.5, 0)+(s)$)}] 0:1.5) arc [x radius=1.5, y radius=0.6, start angle = 0, end angle=180];
\draw ([shift={($(-1.5, 0)+(s)$)}] 0:1.5) arc [x radius=1.5, y radius=0.6, start angle = 0, end angle=180];
\draw ([shift={($(2, 0)+(s)$)}] 0:2) arc [x radius=2, y radius=0.8, start angle = 0, end angle=180];
\draw ([shift={($(-2, 0)+(s)$)}] 0:2) arc [x radius=2, y radius=0.8, start angle = 0, end angle=180];
\draw ([shift={($(-2.5, 0)+(s)$)}] 0:2.5) arc [x radius=2.5, y radius=1, start angle = 0, end angle=180];
\draw ([shift={($(-3, 0)+(s)$)}] 0:3) arc [x radius=3, y radius=1.2, start angle = 0, end angle=180];

\coordinate (s) at (2, 0); 
\draw ([shift={($(0.5, 0)+(s)$)}] 0:0.5) arc [x radius=0.5, y radius=0.2, start angle = 0, end angle=180];
\draw ([shift={($(-0.5, 0)+(s)$)}]0:0.5) arc [x radius=0.5, y radius=0.2, start angle = 0, end angle=180];
\draw ([shift={($(1, 0)+(s)$)}] 0:1) arc [x radius=1, y radius=0.4, start angle = 0, end angle=180];
\draw ([shift={($(-1, 0)+(s)$)}] 0:1) arc [x radius=1, y radius=0.4, start angle = 0, end angle=180];
\draw ([shift={($(1.5, 0)+(s)$)}] 0:1.5) arc [x radius=1.5, y radius=0.6, start angle = 0, end angle=180];
\draw ([shift={($(-1.5, 0)+(s)$)}] 0:1.5) arc [x radius=1.5, y radius=0.6, start angle = 0, end angle=180];
\draw ([shift={($(-2, 0)+(s)$)}] 0:2) arc [x radius=2, y radius=0.8, start angle = 0, end angle=180];
\draw ([shift={($(-2.5, 0)+(s)$)}] 0:2.5) arc [x radius=2.5, y radius=1, start angle = 0, end angle=180];
\draw ([shift={($(-3, 0)+(s)$)}] 0:3) arc [x radius=3, y radius=1.2, start angle = 0, end angle=180];
\draw ([shift={($(-3.5, 0)+(s)$)}] 0:3.5) arc [x radius=3.5, y radius=1.4, start angle = 0, end angle=180];

\coordinate (s) at (3, 0); 
\draw ([shift={($(0.5, 0)+(s)$)}] 0:0.5) arc [x radius=0.5, y radius=0.2, start angle = 0, end angle=180];
\draw ([shift={($(-0.5, 0)+(s)$)}]0:0.5) arc [x radius=0.5, y radius=0.2, start angle = 0, end angle=180];
\draw ([shift={($(1, 0)+(s)$)}] 0:1) arc [x radius=1, y radius=0.4, start angle = 0, end angle=180];
\draw ([shift={($(-1, 0)+(s)$)}] 0:1) arc [x radius=1, y radius=0.4, start angle = 0, end angle=180];
\draw ([shift={($(-1.5, 0)+(s)$)}] 0:1.5) arc [x radius=1.5, y radius=0.6, start angle = 0, end angle=180];
\draw ([shift={($(-2, 0)+(s)$)}] 0:2) arc [x radius=2, y radius=0.8, start angle = 0, end angle=180];
\draw ([shift={($(-2.5, 0)+(s)$)}] 0:2.5) arc [x radius=2.5, y radius=1, start angle = 0, end angle=180];
\draw ([shift={($(-3, 0)+(s)$)}] 0:3) arc [x radius=3, y radius=1.2, start angle = 0, end angle=180];
\draw ([shift={($(-3.5, 0)+(s)$)}] 0:3.5) arc [x radius=3.5, y radius=1.4, start angle = 0, end angle=180];
\draw ([shift={($(-4, 0)+(s)$)}] 0:4) arc [x radius=4, y radius=1.6, start angle = 0, end angle=180];

\coordinate (s) at (4, 0); 
\draw ([shift={($(0.5, 0)+(s)$)}] 0:0.5) arc [x radius=0.5, y radius=0.2, start angle = 0, end angle=180];
\draw ([shift={($(-0.5, 0)+(s)$)}]0:0.5) arc [x radius=0.5, y radius=0.2, start angle = 0, end angle=180];
\draw ([shift={($(-1, 0)+(s)$)}] 0:1) arc [x radius=1, y radius=0.4, start angle = 0, end angle=180];
\draw ([shift={($(-1.5, 0)+(s)$)}] 0:1.5) arc [x radius=1.5, y radius=0.6, start angle = 0, end angle=180];
\draw ([shift={($(-2, 0)+(s)$)}] 0:2) arc [x radius=2, y radius=0.8, start angle = 0, end angle=180];
\draw ([shift={($(-2.5, 0)+(s)$)}] 0:2.5) arc [x radius=2.5, y radius=1, start angle = 0, end angle=180];
\draw ([shift={($(-3, 0)+(s)$)}] 0:3) arc [x radius=3, y radius=1.2, start angle = 0, end angle=180];
\draw ([shift={($(-3.5, 0)+(s)$)}] 0:3.5) arc [x radius=3.5, y radius=1.4, start angle = 0, end angle=180];
\draw ([shift={($(-4, 0)+(s)$)}] 0:4) arc [x radius=4, y radius=1.6, start angle = 0, end angle=180];
\draw ([shift={($(-4.5, 0)+(s)$)}] 0:4.5) arc [x radius=4.5, y radius=1.8, start angle = 0, end angle=180];

\coordinate (s) at (-1, 0); 
\draw ([shift={($(0.5, 0)+(s)$)}] 0:0.5) arc [x radius=0.5, y radius=0.2, start angle = 0, end angle=180];
\draw ([shift={($(-0.5, 0)+(s)$)}]0:0.5) arc [x radius=0.5, y radius=0.2, start angle = 0, end angle=180];
\draw ([shift={($(1, 0)+(s)$)}] 0:1) arc [x radius=1, y radius=0.4, start angle = 0, end angle=180];
\draw ([shift={($(-1, 0)+(s)$)}] 0:1) arc [x radius=1, y radius=0.4, start angle = 0, end angle=180];
\draw ([shift={($(1.5, 0)+(s)$)}] 0:1.5) arc [x radius=1.5, y radius=0.6, start angle = 0, end angle=180];
\draw ([shift={($(-1.5, 0)+(s)$)}] 0:1.5) arc [x radius=1.5, y radius=0.6, start angle = 0, end angle=180];
\draw ([shift={($(2, 0)+(s)$)}] 0:2) arc [x radius=2, y radius=0.8, start angle = 0, end angle=180];
\draw ([shift={($(-2, 0)+(s)$)}] 0:2) arc [x radius=2, y radius=0.8, start angle = 0, end angle=180];
\draw ([shift={($(2.5, 0)+(s)$)}] 0:2.5) arc [x radius=2.5, y radius=1, start angle = 0, end angle=180];
\draw ([shift={($(3, 0)+(s)$)}] 0:3) arc [x radius=3, y radius=1.2, start angle = 0, end angle=180];

\coordinate (s) at (-2, 0); 
\draw ([shift={($(0.5, 0)+(s)$)}] 0:0.5) arc [x radius=0.5, y radius=0.2, start angle = 0, end angle=180];
\draw ([shift={($(-0.5, 0)+(s)$)}]0:0.5) arc [x radius=0.5, y radius=0.2, start angle = 0, end angle=180];
\draw ([shift={($(1, 0)+(s)$)}] 0:1) arc [x radius=1, y radius=0.4, start angle = 0, end angle=180];
\draw ([shift={($(-1, 0)+(s)$)}] 0:1) arc [x radius=1, y radius=0.4, start angle = 0, end angle=180];
\draw ([shift={($(1.5, 0)+(s)$)}] 0:1.5) arc [x radius=1.5, y radius=0.6, start angle = 0, end angle=180];
\draw ([shift={($(-1.5, 0)+(s)$)}] 0:1.5) arc [x radius=1.5, y radius=0.6, start angle = 0, end angle=180];
\draw ([shift={($(2, 0)+(s)$)}] 0:2) arc [x radius=2, y radius=0.8, start angle = 0, end angle=180];
\draw ([shift={($(2.5, 0)+(s)$)}] 0:2.5) arc [x radius=2.5, y radius=1, start angle = 0, end angle=180];
\draw ([shift={($(3, 0)+(s)$)}] 0:3) arc [x radius=3, y radius=1.2, start angle = 0, end angle=180];
\draw ([shift={($(3.5, 0)+(s)$)}] 0:3.5) arc [x radius=3.5, y radius=1.4, start angle = 0, end angle=180];

\coordinate (s) at (-3, 0); 
\draw ([shift={($(0.5, 0)+(s)$)}] 0:0.5) arc [x radius=0.5, y radius=0.2, start angle = 0, end angle=180];
\draw ([shift={($(-0.5, 0)+(s)$)}]0:0.5) arc [x radius=0.5, y radius=0.2, start angle = 0, end angle=180];
\draw ([shift={($(1, 0)+(s)$)}] 0:1) arc [x radius=1, y radius=0.4, start angle = 0, end angle=180];
\draw ([shift={($(-1, 0)+(s)$)}] 0:1) arc [x radius=1, y radius=0.4, start angle = 0, end angle=180];
\draw ([shift={($(1.5, 0)+(s)$)}] 0:1.5) arc [x radius=1.5, y radius=0.6, start angle = 0, end angle=180];
\draw ([shift={($(2, 0)+(s)$)}] 0:2) arc [x radius=2, y radius=0.8, start angle = 0, end angle=180];
\draw ([shift={($(2.5, 0)+(s)$)}] 0:2.5) arc [x radius=2.5, y radius=1, start angle = 0, end angle=180];
\draw ([shift={($(3, 0)+(s)$)}] 0:3) arc [x radius=3, y radius=1.2, start angle = 0, end angle=180];
\draw ([shift={($(3.5, 0)+(s)$)}] 0:3.5) arc [x radius=3.5, y radius=1.4, start angle = 0, end angle=180];
\draw ([shift={($(4, 0)+(s)$)}] 0:4) arc [x radius=4, y radius=1.6, start angle = 0, end angle=180];

\coordinate (s) at (-4, 0); 
\draw ([shift={($(0.5, 0)+(s)$)}] 0:0.5) arc [x radius=0.5, y radius=0.2, start angle = 0, end angle=180];
\draw ([shift={($(-0.5, 0)+(s)$)}]0:0.5) arc [x radius=0.5, y radius=0.2, start angle = 0, end angle=180];
\draw ([shift={($(1, 0)+(s)$)}] 0:1) arc [x radius=1, y radius=0.4, start angle = 0, end angle=180];
\draw ([shift={($(1.5, 0)+(s)$)}] 0:1.5) arc [x radius=1.5, y radius=0.6, start angle = 0, end angle=180];
\draw ([shift={($(2, 0)+(s)$)}] 0:2) arc [x radius=2, y radius=0.8, start angle = 0, end angle=180];
\draw ([shift={($(2.5, 0)+(s)$)}] 0:2.5) arc [x radius=2.5, y radius=1, start angle = 0, end angle=180];
\draw ([shift={($(3, 0)+(s)$)}] 0:3) arc [x radius=3, y radius=1.2, start angle = 0, end angle=180];
\draw ([shift={($(3.5, 0)+(s)$)}] 0:3.5) arc [x radius=3.5, y radius=1.4, start angle = 0, end angle=180];
\draw ([shift={($(4, 0)+(s)$)}] 0:4) arc [x radius=4, y radius=1.6, start angle = 0, end angle=180];
\draw ([shift={($(4.5, 0)+(s)$)}] 0:4.5) arc [x radius=4.5, y radius=1.8, start angle = 0, end angle=180];
\end{tikzpicture}
\end{center}
\caption{All vertices are connected to all other vertices, for $d=1$}
\label{pic_all_connect}
\end{figure}
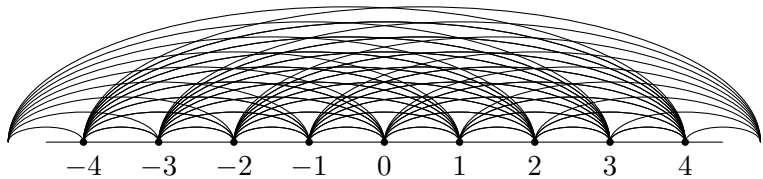
\end{Example}

Let us now be more precise about the content of this paper.
As already mentioned, we investigate the spectral theory of Schr\"{o}dinger operators on discrete graphs. Those graphs are obtained by perturbing an initial topological crystal, either at the level of edges (addition of an arbitrary number of them, removal of a finite number of them, change of weights), or at the level of the vertices (change of weights). 
Perturbation by the addition of a multiplication operator is also allowed.
These investigations on a self-adjoint operator $H$ are performed by comparing it
to an unperturbed self-adjoint operator $H_0$ acting on the initial topological crystal.
Powerful tools have been developed for investigating $H_0$, and our strategy 
is to adapt these tools for investigating $H$.
More precisely, if the difference between $H$ and $H_0$ is small and regular
in a suitable sense, then some properties of $H$ can be deduced from similar
properties of $H_0$.
In particular, by applying perturbative techniques, one can deduce the following results: 
the spectrum of $H$ consists of absolutely continuous spectrum, of a finite number (possibly empty) of eigenvalues of infinite multiplicity, and of eigenvalues of finite multiplicity which can accumulate only at a discrete set of thresholds.
In addition, one proves the existence and the completeness of the local
wave operators for the pair of operators $(H,H_0)$.
Note that these investigations generalize the results obtained in \cite{PR}
which were obtained in the framework of general topological crystals. 

As already mentioned, operators acting on graphs have been extensively studied. Among all corresponding papers, we list only those which are clearly linked to our investigations.
First of all, our main reference about topological crystals is the book \cite{Su12}.
On such discrete structures, it is well-known that periodic operators have a band structure with at most a finite number of eigenvalues of infinite multiplicity, see for example \cite{HN09,KS14, KS15}. 
The next step is to study what happens when these periodic Schr\"{o}dinger operators are perturbed. Two main types of perturbations can be considered.

The first one consists in adding a potential that decays at infinity as a short-range function
or as a long-range function. For general topological crystals, we refer to \cite{PR} and to the references mentioned therein. 
For specific graphs, these types of perturbations have also been studied in greater detail.
For example, the case of $\Z^{d}$ and graphene have been fully investigated in \cite{BS99, Ta19_2} and \cite{Ta19_1} respectively. 
As a related work, \cite{KM} provides estimates for the unitary group and the resolvent of the discrete Laplace operator on $\Z^{d}$, from which the authors infer some results for the spectral and the scattering theory of perturbed operators by potentials $V$ vanishing at infinity. 

The second type of perturbations is the modification of the graph itself. 
Perturbations corresponding to changing the weights of the graph have been investigated for example in \cite{GK18, PR}. 
For a perturbation transforming the graph structure, 
\cite{AIM15} studies spectral properties of Schr\"{o}dinger operators on perturbed periodic lattice including square, triangular, diamond, and kagome lattices, 
but the perturbations considered there are only compactly supported and some implicit conditions on the Floquet-Bloch variety are assumed.  
Note that some related results on the inverse scattering problem with compactly supported perturbations are available for some specific graphs in \cite{An13, AIM18, IK12, IM15}. 
For a non-compact perturbation, \cite{SS17} studies the stability of their essential spectrum. 
For Schr\"{o}dinger operators on periodic graphs perturbed by guides, graphs which are periodic in some directions and finite in other ones, we refer to \cite{KS17_1, KS17_2}. 

Let us now describe the sections of this paper.
In Section \ref{sec_main_result}, we describe the framework of our investigations and provide our main results. Note that topological crystals and more general graphs are thoroughly presented in this section.
The technical tools are introduced in Section \ref{section:other}.
This material is mainly borrowed from the paper \cite{PR}. In particular, 
we review the notion of analytically fibered operator, and recall that $H_0$ is unitarily equivalent to such an operator. Mourre theory and a suitable conjugate operator are also briefly introduced, and the Mourre estimate for $H_0$ is recalled. 
Note that another version of Mourre theory applied to discrete periodic operators has also 
been introduced in \cite{MRT07}. 
The proofs of all results, including the regularity of the difference between $H_{0}$ and $H$ are provided in Section \ref{sec_pert}. The examples mentioned in this introduction are also fully treated.  

As a final remark, let us emphasize one interest of the framework of topological crystals: The regularity of a graph, given by an action of $\Z^d$, is independent of the dimension on which the graph is naturally represented. In particular, these two dimensions can be different, as illustrated in the following example: The graph is naturally represented in $\R^3$ while the group acting is only $\Z$.

\begin{Example}[Toblerone \textregistered]\label{ex_Toblerone}
We consider the $1$ dimensional topological crystal in Figure \ref{original_toblerone}, 
and add infinitely many edges connecting one vertex to all other vertices. 
We denote by $x_{0}$ this special vertex. 
The set of added edges is denoted by $F$, which corresponds to the set of bold edges in Figure \ref{added_toblerone}. 
If we denote by $x_0$, $y_0$ and $z_0$ the three vertices of the triangle (section) 
containing $x_0$, then all other vertices can be naturally indexed by $x_\mu$, $y_\mu$ and $z_\mu$ for $\mu\in \Z$. 
We also fix $m(x)=1$ for any vertices $x$ and
$m(\e)=1$ for all initial edges $\e$. For any $\e \in  F$, with 
endpoints $x_{0}$ and $k_\mu$ with $k\in \{x,y,z\}$, we assume that  
$m(\e) \leq C (1+|\mu|)^{\alpha}$
for some $\alpha<-3$ and some constant $C$ independent of $\e$. 
\begin{figure}[h]
\begin{center}
\subfigure[Original Toblerone]{
\tikzset{node/.style={circle,fill=black,inner sep=1pt}}
\tikzset{node1/.style={circle,fill=black,inner sep=1.5pt}}
\begin{tikzpicture}[scale=0.75]
\coordinate[node] (a1) at (0,0);
\coordinate[node] (b1) at (-1.5,0.5);
\coordinate[node] (c1) at (-0.5, 2.5);

\coordinate[node] (a2) at (2, 0.25);
\coordinate[node] (b2) at (0.5, 0.75);
\coordinate[node] (c2) at (1.5, 2.75);

\coordinate[node] (a3) at (4, 0.5);
\coordinate[node] (b3) at (2.5, 1);
\coordinate[node] (c3) at (3.5, 3);

\coordinate[node] (a4) at (6, 0.75);
\coordinate[node] (b4) at (4.5, 1.25);
\coordinate[node] (c4) at (5.5, 3.25);

\draw[] (a1)--(b1);
\draw[] (b1)--(c1);
\draw[] (c1)--(a1);

\draw[] (a2)--(b2);
\draw[] (b2)--(c2);
\draw[] (c2)--(a2);

\draw[] (a3)--(b3);
\draw[] (b3)--(c3);
\draw[] (c3)--(a3);

\draw[] (a4)--(b4);
\draw[] (b4)--(c4);
\draw[] (c4)--(a4);

\coordinate (x0) at (-1, -0.125);
\coordinate (x1) at (7, 0.875);
\draw[] (x0)--(x1);

\coordinate (y0) at (-2.5, 0.375);
\coordinate (y1) at (5.5, 1.375);
\draw[] (y0)--(y1);

\coordinate (z0) at (-1.5, 2.375);
\coordinate (z1) at (6.5, 3.375);
\draw[] (z0)--(z1);
\end{tikzpicture}
\label{original_toblerone}
}
\hfill
\subfigure[Toblerone with added edges]{
\tikzset{node/.style={circle,fill=black,inner sep=1pt}}
\tikzset{node1/.style={circle,fill=black,inner sep=1.5pt}}
\begin{tikzpicture}[scale=0.75]
\coordinate[node] (a1) at (0,0);
\coordinate[node] (b1) at (-1.5,0.5);
\coordinate[node] (c1) at (-0.5, 2.5);

\coordinate[node1, label=below: {$x_{0}$}] (a2) at (2, 0.25);
\coordinate[node, label=above left: {$z_{0}$}] (b2) at (0.5, 0.75);
\coordinate[node, label=above: {$y_{0}$}] (c2) at (1.5, 2.75);

\coordinate[node] (a3) at (4, 0.5);
\coordinate[node] (b3) at (2.5, 1);
\coordinate[node] (c3) at (3.5, 3);

\coordinate[node] (a4) at (6, 0.75);
\coordinate[node] (b4) at (4.5, 1.25);
\coordinate[node] (c4) at (5.5, 3.25);

\draw[] (a1)--(b1);
\draw[] (b1)--(c1);
\draw[] (c1)--(a1);

\draw[] (a2)--(b2);
\draw[] (b2)--(c2);
\draw[] (c2)--(a2);

\draw[] (a3)--(b3);
\draw[] (b3)--(c3);
\draw[] (c3)--(a3);

\draw[] (a4)--(b4);
\draw[] (b4)--(c4);
\draw[] (c4)--(a4);

\coordinate (x0) at (-1, -0.125);
\coordinate (x1) at (7, 0.875);
\draw[] (x0)--(x1);

\coordinate (y0) at (-2.5, 0.375);
\coordinate (y1) at (5.5, 1.375);
\draw[] (y0)--(y1);

\coordinate (z0) at (-1.5, 2.375);
\coordinate (z1) at (6.5, 3.375);
\draw[] (z0)--(z1);

\draw[thick, rounded corners=17pt] (a2)--($(a2)!.75!(a3)!1!-90:(a3)$)--(a3);
\draw[thick, rounded corners=17pt] (a1)--($(a1)!.75!(a2)!1!-90:(a2)$)--(a2);
\draw[thick, rounded corners=30pt] (a2)--($(a2)!.7!(a4)!0.8!-90:(a4)$)--(a4);

\draw[thick, rounded corners=15pt] (a2)--($(a2)!.5!(b1)!.1!90:(b1)$)--(b1);
\draw[thick, rounded corners=10pt] (a2)--($(a2)!.75!(b2)!1!90:(b2)$)--(b2);
\draw[thick, rounded corners=10pt] (a2)--($(a2)!.25!(b3)!.3!-90:(b3)$)--(b3);
\draw[thick, rounded corners=20pt] (a2)--($(a2)!.5!(b4)!.2!-90:(b4)$)--(b4);

\draw[thick, rounded corners=30pt] (a2)--($(a2)!.55!(c1)!.3!90:(c1)$)--(c1);
\draw[thick, rounded corners=30pt] (a2)--($(a2)!.6!(c2)!1!-90:(c2)$)--(c2);
\draw[thick, rounded corners=35pt] (a2)--($(a2)!.55!(c3)!.8!-80:(c3)$)--(c3);
\draw[thick, rounded corners=40pt] (a2)--($(a2)!.6!(c4)!.5!-90:(c4)$)--(c4);
\end{tikzpicture}
\label{added_toblerone}
}
\end{center}
\caption{A $1$ dimensional topological crystal and its perturbation}
\label{toblerone}
\end{figure}
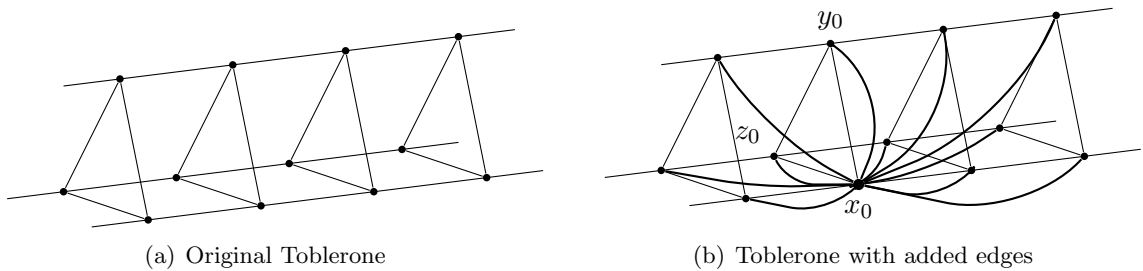
\end{Example}

\section{Framework and main result}\label{sec_main_result}
\setcounter{equation}{0}

In this section we first introduce the necessary information about general graphs and 
topological crystals, and then state our main result. Its proof will be provided in Section
\ref{sec_pert}.

\subsection{General graphs}\label{sec_gen}

A graph $X=\big(V(X),E(X)\big)$ is composed of a set $V(X)$ of vertices and a set 
$E(X)$ of unoriented edges. Multiple edges and loops are accepted.
Generically we shall use the notation $x,y$ for elements of $V(X)$, and $\e$ for elements of $E(X)$.
If both $V(X)$ and $E(X)$ are finite sets, the graph $X$ is said to be finite.

From the set of unoriented edges $E(X)$ of the graph $X$ we construct the set $\A(X)$ of oriented edges by defining, for any unoriented edge between $x$ and $y$, one oriented edge from $x$ to $y$ and one oriented edge from $y$ to $x$.
The elements of $\A(X)$ are also denoted by $\e$.
The origin vertex of such an oriented edge $\e$ is denoted by $o(\e)$, the terminal one by $t(\e)$,
and $\overline{\e}$ corresponds to the edge obtained from $\e$ by interchanging the vertices, \ie~$o(\overline{\e})=t(\e)$ and $t(\overline{\e})=o(\e)$.

For a vertex $x\in V(X)$ we set $E(X)_x:=\{\e\in E(X)\mid x \hbox{ is an endpoint of }\e\}$.
If $E(X)_x$ is finite for every $x\in V(X)$ we say that $X$ is locally finite.
Similarly, for $x\in V(X)$ we set $A(X)_x:=\{\e\in\A(X)\mid o(\e)=x\}$.
If there is no ambiguity about the graph, we shall simply write $E_x$ for $E(X)_x$ and
$A_x$ for $A(X)_x$.

By a measure $m$ on a graph $X$, we mean a function $m$ defined on vertices and on unoriented edges satisfying $m(x)>0$ and $m(\e)\geq 0$ for any $x\in V(X)$ 
and $\e\in E(X)$.
Note that measures in \cite{PR} were considered strictly positive, while here we assume
strict positivity on vertices, but allow the value $0$ on edges.
The measure on an oriented edge is defined by its value on the corresponding unoriented edge. As a consequence, the measure satisfies $m(\e)=m(\overline{\e})$.

Consider now the set 
$$
C_c(X):=\big\{f:V(X)\to \C\mid f(x)=0 \hbox{ except for a finite number of } x\in V(X)\big\},
$$
and define the degree function 
\begin{equation}\label{eq_def_deg}
\deg_{m}: V(X)\to[0,\infty), \quad \deg_{m} (x):=\sum_{\e\in\A_{x}}\frac{m(\e)}{m(x)}.
\end{equation}
If this function is bounded, then the Laplace operator given by
\begin{equation*}
\[\Delta(X,m)f\](x):=\sum_{\e\in\A_x}\frac{m(\e)}{m(x)}\big(f\big(t(\e)\big)-f(x)\big) \qquad  \forall f\in C_c(X),
\end{equation*}
extends continuously to a bounded and self-adjoint operator in the Hilbert space
\begin{equation*}
\ell^2(X,m):=\Big\{f:V(X)\to \C \mid \lp f\rp^2:=\sum_{x\in V(X)}\;\!m(x)|f(x)|^2<\infty\Big\}
\end{equation*}
endowed with the scalar product
\begin{equation*}
\langle f,g\rangle := \sum_{x\in V(X)}\;\!m(x)f(x)\;\!\overline{g(x)}\qquad \forall f,g\in \ell^2(X,m).
\end{equation*}
Note that the boundedness of $\Delta(X,m)$ has been proved in \cite[Thm.~2.4]{Keller}
and that the self-adjointness follows from a Green's formula, as proved in 
\cite[Lem.~4.7]{HK}. Let us also stress that these results do not assume local finiteness
of the graph, only the boundedness of the degree function $\deg_m$. This latter condition
will be assumed throughout the paper.
  
Let us finally consider a bounded function $R:V(X)\to \R$, and identify it with a multiplication operator in $\ell^2(X,m)$. Then, we end up with the following bounded and
self-adjoint operator $H$ which will be our main object of interest, namely
\begin{equation}\label{full}
H:=-\Delta (X,m)+ R.
\end{equation}
In fact, such a general operator will be considered later as a perturbation of 
a periodic operator on a topological crystal. 

\subsection{Topological crystals}\label{basic_topological_crystals}

In this section we provide the definition of a topological crystal
and define some related notions. Most of this material is directly borrowed from
\cite{PR} and \cite{R}.

A morphism $\omega: X\to \XX$ between two graphs $X$ and $\XX$ is composed of two maps $\omega:V(X)\to V(\XX)$
and $\omega:E(X)\to E(\XX)$ such that it preserves the adjacency relations between vertices and edges,
namely if $\e$ is an edge in $X$ between the vertices $x$ and $y$, then $\omega(\e)$ is an edge in $\XX$ between the vertices $\omega(x)$ and $\omega(y)$.
Clearly, any morphism can be extended to a map sending oriented edges of $\A(X)$
to oriented edges of $\A(\XX)$. For this extension we keep the convenient notation $\omega:\A(X)\to \A(\XX)$.
An isomorphism is a morphism that is a bijection on the vertices and on the edges.
The group of isomorphisms of a graph $X$ into itself is denoted by $\Aut(X)$.

A morphism $\omega:X\to\XX$ between two graphs is said to be a covering map if
\begin{enumerate}
\item[(i)] $\omega:V(X)\to V(\XX)$ is surjective,
\item[(ii)] for all $x\in V(X)$, the restriction $\omega|_{E(X)_x}:E(X)_x\to E(\XX)_{\omega(x)}$ is a bijection.
\end{enumerate}
In that case we say that $X$ is a covering graph over the base graph $\XX$.
For such a covering, we define the transformation group of the covering as the subgroup of $\Aut(X)$,
denoted by $\Gamma$ and with the multiplicative notation, 
such that for every $\mu\in\Gamma$ the equality $\omega\circ\mu=\omega$ holds.
We now define a topological crystal, and refer to \cite[Sec. 6.2]{Su12} for more details.

\begin{Definition}\label{topocrystal}
A $d$-dimensional topological crystal is a quadruplet $(X,\XX,\omega,\Gamma)$ such that:
\begin{enumerate}
\item[(i)] $X$ and $\XX$ are graphs, with $\XX$ finite,
\item[(ii)] $\omega:X\to \XX$ is a covering map,
\item[(iii)] The transformation group $\Gamma$ of $\omega$ is isomorphic to $\Z^d$,
\item[(iv)] $\omega$ is regular, \ie~for every $x$, $y\in V(X)$ satisfying $\omega(x)=\omega(y)$ there exists $\mu\in\Gamma$ such that $x=\mu y$.
\end{enumerate}
\end{Definition}

For simplicity, we assume that topological crystals have no multiple edges, and we shall just say that $X$ is a topological crystal if it admits a $d$-dimensional topological crystal structure $(X,\XX,\omega,\Gamma)$.
Note that we use the multiplicative notation for the group law in the abstract setting,
but the additive notation when dealing explicitly with $\Z^d$.
Note also that all topological crystal are locally finite, with an upper bound for the number of elements in $E(X)_x$ independent of $x$.
Indeed, the local finiteness and the fixed upper bound follow from the definition of a covering and the finiteness of $\XX$.

Topological crystals have been extensively studied in the monograph \cite{Su12} to which we refer for many examples.
Let us also mention \cite{AIM15} in which one can find square, triangular, hexagonal, and diamond periodic graphs.
In reference \cite{KS14} body-centered cubic and face-centered cubic periodic graphs have been studied,
while armchair graph is presented in \cite{BK10}.
We also refer to \cite[Rem.~3.1]{PR} for an explicit procedure generating an infinite number  of topological crystals
$(X,\XX,\omega,\Gamma)$ once a small graph $\XX$ has been chosen.

Let us add a few definitions related to a topological crystal $(X,\XX,\omega,\Gamma)$.
The notation $x$, resp.~$\xx$, will be used for the elements of $V(X)$, resp.~of $V(\XX)$,
and accordingly the notation $\e$, resp.~$\ee$, will be used for the elements of $E(X)$, resp.~of $E(\XX)$.
It follows from the assumptions in Definition \ref{topocrystal} that we can identify $V(\XX)$ as a subset of $V(X)$ by choosing a representative of each orbit.
Namely, since $V(\XX)=\{\xx_1,\dots,\xx_n\}$ for some $n\in\N$,
we choose $\{x_1,\dots,x_n\}\subset V(X)$ such that $\omega(x_j)=\xx_j$ for any $j\in \{1,\dots,n\}$.
For shortness we also use the notation $\check{x}:=\omega(x)\in V(\XX)$ for any $x\in V(X)$,
and reciprocally for any $\xx\in V(\XX)$ we write $\hat{\xx}\in \{x_1,\dots,x_n\}$ for the unique element $x_j$ in this set such that $\omega(x_j)=\xx$.

As a consequence of the previous construction we can also identify $\A(\XX)$ as a subset of $\A(X)$.
More precisely, we identify $\A(\XX)$ with $\cup_{j=1}^n \A_{x_j}\subset\A(X)$ and use notations similar to the previous ones:
For any $\e\in \A(X)$ one sets $\check{\e}:=\omega(\e)\in \A(\XX)$, and for any $\ee\in \A(\XX)$ one sets $\hat{\ee}\in \cup_{j=1}^n\A_{x_j}$
for the unique element in $\cup_{j=1}^n\A_{x_j}$ such that $\omega(\hat{\ee})=\ee$.
Let us stress that these identifications and notations depend only on the initial choice of $\{x_1,\dots,x_n\}\subset V(X)$.

We have now enough notation for defining the entire part of a vertex $x$ as the map $\lfloor\,\cdot\,\rfloor :V(X)\to\Gamma$ satisfying
\begin{equation*}
\lfloor x\rfloor \widehat{\check{x}}=x.
\end{equation*}
Similarly, the entire part of an edge is defined as the map $\lfloor\,\cdot\,\rfloor :\A(X)\to\Gamma$ satisfying
\begin{equation*}
\lfloor \e\rfloor \widehat{\check{\e}}=\e.
\end{equation*}
The existence of this function $\lfloor \, \cdot\,\rfloor $ follows from the assumption (iv)
of Definition \ref{topocrystal} on the regularity of a topological crystal.
One easy consequence of the previous construction is that the equality
$\lfloor \e\rfloor =\lfloor o(\e)\rfloor $ holds for any $\e\in \A(X)$.

For later use, let us also define the map
\begin{equation*}
\eta:\A(X)\to\Gamma, \quad \eta(\e):=\lfloor t(\e)\rfloor \lfloor o(\e)\rfloor ^{-1}
\end{equation*}
and call $\eta(\e)$ the index of the edge $\e$. For any $\mu\in\Gamma$ we then infer that
\begin{equation*}
\eta(\mu\e)=\lfloor t(\mu\e)\rfloor \lfloor o(\mu\e)\rfloor ^{-1}=\mu\lfloor t(\e)\rfloor \mu^{-1}\lfloor o(\e)\rfloor ^{-1}=\eta(\e).
\end{equation*}
This periodicity enables us to define unambiguously $\eta:\A(\XX)\to \Gamma$ by the relation $\eta(\ee):=\eta(\hat{\ee})$
for every $\ee\in\A(\XX)$. Again, this index on $\A(\XX)$ depends only on the initial choice $\{x_1,\dots,x_n\}\subset V(X)$
and could not be defined by considering only $\A(\XX)$.

Let us now come back to operators acting on a topological crystal. 
In this setting, we consider a $\Gamma$-periodic measure $m_0$ and a $\Gamma$-periodic function $R_0: V(X)\to \R$.
The periodicity means that for every $\mu\in\Gamma$, $x\in V(X)$ and $\e \in E(X)$ we have $m_0(\mu x)=m_0(x)$, $m_0(\mu \e) = m_0(\e)$ and $R_0(\mu x)=R_0(x)$.
In this framework, the degree function $\deg_{m_0}$ introduced in \eqref{eq_def_deg}
is clearly bounded.
Therefore, a periodic Schr\"odinger operator defined by
\begin{equation}\label{h0vertices}
H_0:=-\Delta (X,m_0)+R_0
\end{equation}
is a bounded and self-adjoint operator in the Hilbert space $\ell^2(X,m_0)$.
Such an operator corresponds to our unperturbed system.

\subsection{Perturbations of topological crystals}

In this section we introduce perturbations of topological crystals by adding
and/or removing edges. 
The framework is the following:
Let $(X,\XX,\omega,\Gamma)$ be a topological crystal, with $X=\big(V(X),E(X)\big)$, 
and let $m_0$ be a $\Gamma$-periodic measure on $X$. 
We now consider the addition and/or the elimination of edges.
For the addition, let $F_+$ be a possibly infinite set of unoriented new edges between arbitrary vertices of $X$. 
For the elimination, we consider a finite subset $F_-$ of edges of $E(X)$. 
Without loss of generality, we assume that $F_+$ and $F_-$ do not contain multiple edges.
We then obtain a new graph $\X=\big(V(\X),E(\X)\big)$ given by
$V(\X):=V(X)$ and $E(\X):= \big(E(X)\setminus F_-\big)\cup F_+$.
In general this graph is no more a topological crystal.
For this new graph, the set of oriented edges is denoted by $A(\X)$.
Let us also define $A(F_+)$ and $A(F_-)$: the first set corresponds to the sets of oriented edges based on $F_+$, with $A(F_+)\subset A(\X)$, while $A(F_-)\subset A(X)$ is the set of oriented edges based on $F_-$. Note that $A(F_-)$ is not included in $A(\X)$ in general.

We then consider a measure $m$ on $\X$ with $\deg_m$ bounded. The corresponding Laplace operator $\Delta(\X, m)$ is then self-adjoint and bounded in $\ell^2(\X,m)$.
Subsequently, the measure $m$ restricted to the edges in $E(X)\setminus F_-$ will correspond to a perturbation of $m_0$, 
and similarly the measure $m$ on the vertices in $V(X)$ will be a perturbation of the measure $m_0$.
It will also be useful to introduce a partial degree function, namely
\begin{equation}\label{def_deg_F}
\deg_{F_+}:V(\X)\to [0,\infty), \quad \deg_{F_+}(x):= \sum_{\e\in A(F_+)_{x}}  \frac{m(\e)}{m(o(\e))}.
\end{equation}
Clearly, $\deg_{F_+}\leq \deg_m$, and this function is bounded since the function $\deg_m$ is assumed to be bounded.

Before stating our main result, 
let us still mention that the isomorphism between $\Gamma$ and $\Z^d$ allows us to borrow the Euclidean norm $|\cdot|$ of $\Z^d$ and to
endow $\Gamma$ with it. As a consequence of this construction, the notations $|\lfloor x\rfloor |$ and $|\lfloor \e\rfloor |$ are well-defined,
and the notion of rate of convergence towards infinity is available.
Also, since no vertex has been added or eliminated between the original graph and the perturbed one, we can introduce a unitary transformation $\J : \ell^{2}(X, m_0)\to\ell^{2}(\X, m)$ given by  
\begin{equation}\label{eq_map_J}
[\J f](x) := \left( \frac{m_0(x)}{m(x)} \right)^{\frac{1}{2}}f(x), \quad f\in\ell^{2}(X, m_0). 
\end{equation}
Note that in \cite{PR}, the notation $\J$ was used for our map $\J^*$. 

\begin{Theorem}\label{thm_main}
Let $X$ be a topological crystal, endowed with a $\Gamma$-periodic measure $m_0$ 
and a $\Gamma$-periodic function $R_0$.
Let $F_+$ be a possibly infinite set of unoriented new edges, let $F_{-}$ be a finite subset of $E(X)$, and consider the graph $\X=\big(V(\X),E(\X)\big)$ given by
$V(\X):=V(X)$ and $E(\X):= \big(E(X)\setminus F_-\big)\cup F_+$.
Consider a measure $m$ on $\X$ with $\deg_m$ bounded, and assume that $m$ satisfies 
\begin{enumerate}
\item[(i)] Decay of perturbation on pre-existing edges:
\begin{equation}\label{eq_Cm1}
\int_{1}^{\infty}\d\lambda \sup_{\substack{\e\in E(X)\setminus F_- \\ \lambda<|\lfloor \e\rfloor |<2\lambda}}\left|\frac{m(\e)}{m(o(\e))}-\frac{m_0(\e)}{m_0(o(\e))}\right|<\infty,
\end{equation}
\item[(ii)] Decay of degree function on new edges:
\begin{equation}\label{eq_Cm2}
\int_{1}^{\infty} \d \lambda \sup_{\lambda < | \lfloor x \rfloor| < 2\lambda} 
\deg_{F_+}(x) < \infty,
\end{equation}
\item[(iii)] Decay of global new connectivity:
\begin{equation}\label{eq_Cm3}
\int_{1}^{\infty} \d \lambda \sup_{ x\in V(\X) } 
\sqrt{\sum_{\substack{\e \in A(F_+)_{x} \\ \lambda\leq | \lfloor t(\e) \rfloor| \leq2 \lambda}}  \frac{m(\e)}{m(o(\e))}}  < \infty.
\end{equation}
\end{enumerate}
Consider also $R:V(\X)\to \R$ satisfying the decay condition
\begin{equation}\label{eq_Cm4}
\int_{1}^{\infty}\d\lambda \sup_{\lambda<|\lfloor x\rfloor |<2\lambda}\left|R(x)-R_0(x)\right|<\infty.
\end{equation}
Let finally $H_0$ and $H$ be the self-adoint operators defined by \eqref{h0vertices} and \eqref{full} respectively.
Then, there exists a discrete set $\tau\subset\R$ such that for every closed interval $I\subset\R\setminus\tau$ the following assertions hold:
\begin{enumerate}
\item $H_0$ has no eigenvalue in $I$, and $H$ has at most a finite number of eigenvalues in $I$, each of them being of finite multiplicity,
\item $\sigma_{sc}(H_0)\cap I = \sigma_{sc}(H)\cap I=\emptyset$.
\end{enumerate}
If the following additional condition also holds for some $s>1/2$~:
\begin{equation}\label{eq_Cm5}
\sup_{x\in V(\X)}\langle \lfloor x \rfloor \rangle^{2s}
\sum_{\e \in A(F_+)_{x}} \frac{m(\e)}{m(x)} \langle \lfloor t(\e) \rfloor  \rangle^{2s}<\infty,
\end{equation}
then the local wave operators
$$
W\pm \equiv W_{\pm}(H,H_0;\J,I):=s-\lim_{t\to\pm\infty}e^{iH t}\J e^{-iH_0 t}E^{H_0}(I)
$$
exist and satisfy $\Ran (W_-)=\Ran( W_+)=E^{H}_{ac}(I)\ell^2(\X,m)$.
\end{Theorem}

The hypothesis \eqref{eq_Cm1} and \eqref{eq_Cm4} are usually referred to as a short-range type of decay.
In particular it is satisfied for functions that decay faster than $C(1+|\lfloor x\rfloor |)^{-1-\epsilon}$ for $\epsilon>0$ and some constant $C$ independent of $x$.
It is worth mentioning that condition \eqref{eq_Cm1} is quite general and is automatically satisfied if the difference $m-m_0$ itself satisfies a
short-range type of decay. For example if we assume that $|m(\e)-m_0(\e)|\leq C (1+|\lfloor \e\rfloor |)^{-1-\epsilon}$ and
$|m(x)-m_0(x)|\leq C' (1+|\lfloor x\rfloor |)^{-1-\epsilon}$, then \eqref{eq_Cm1} is satisfied.
Note also that the conditions \eqref{eq_Cm2} and \eqref{eq_Cm3} are trivially satisfied if $F_+$ is 
a finite set. On the other hand, if $F_+$ is infinite, the two conditions prescribe 
precisely the necessary decay of the measure on the new edges.
In this case, the conditions allow the addition of an infinite number of edges, both locally and at infinity.
Note finally that the additional condition \eqref{eq_Cm5} is necessary because of 
the non-locality of our perturbations: when dealing with multiplicative perturbations, 
this condition often follows from \eqref{eq_Cm2}.

\section{Analyticity and Mourre theory}\label{section:other}
\setcounter{equation}{0}

In this section, we briefly recall a few tools introduced in \cite{PR} for the study of 
topological crystals and their perturbations, and refer to this reference for the details.
The framework is a topological crystal  $(X,\XX,\omega,\Gamma)$, 
a $\Gamma$-periodic measure $m_0$ and a $\Gamma$-periodic function $R_0$.

\subsection{Analyticity of the periodic operator}

The aim of this section is to introduce another representation of the operator $H_0$,
more suitable for further investigations.
More precisely, we shall obtain that $H_0$ is unitarily equivalent to an analytically fibered operator.
We provide below the simplest definition of such an operator, and refer to \cite{GN98} and \cite[Sec.~XIII.16]{RS4} for general information.
Note that from now on we shall use the notation $\T^d$  for the $d$-dimensional torus, \ie~for $\T^d=\R^{d}/\Z^{d}$,
with the inherited local coordinates system and differential structure.
We shall also use the notation $M_n(\C)$ for the $n\times n$ matrices over~$\C$.

\begin{Definition}\label{analiticallyfibered}
In the Hilbert space $L^2(\T^d;\C^n)$, a bounded analytically fibered operator
corresponds to a multiplication operator defined by a real analytic map $h :\T^d\to M_n(\C)$.
\end{Definition}

In order to show that the periodic operator $H_0$ fits into this framework,
some identifications are necessary.
First of all, since $\Gamma$ is isomorphic to $\Z^d$, as stated in the point (iii) of Definition \ref{topocrystal},
we know that its dual group $\hat{\Gamma}$ is isomorphic to $\T^d$. In fact, we consider that a basis of $\Gamma$ is chosen
and then identify $\Gamma$ with $\Z^{d}$, and accordingly $\hat{\Gamma}$ with $\T^d$.
As a consequence of these identifications we set $\xi\cdot\mu = \sum_{j=1}^d\xi_j\mu_j$ for $\xi \in \T^d$ and $\mu\in \Z^d$, and define the Fourier transform for any $f\in \ell^1(\Z^d)$ by 
\begin{equation}\label{def_Fourier}
[\F f](\xi)=\sum_{\mu\in \Z^d} e^{-2\pi i\,\xi\cdot\mu} f(\mu).
\end{equation}
Its inverse is given by 
$$
[\F^* u](\mu)=\int_{\T^d}\d\xi\;\! e^{2\pi i\,\xi\cdot \mu}u(\xi),
$$ 
with $\d \xi$ the usual measure on  $\T^d$.
Note that another consequence of these identifications is the use of the additive notation for the composition of two elements of $\Z^d$,
instead of the multiplicative notation employed until now for the composition in $\Gamma$.

The second necessary identification is between $\ell^2(\XX)$ and $\C^n$. 
Observe firstly that because of the periodicity of the measure $m_0$, 
this measure is also well-defined on $\XX$
by the relation $m_0(\xx):=m_0(\hat{\xx})$ and $m_0(\ee):=m_0(\hat{\ee})$.
For simplicity, we keep the same notation for this measure on $\XX$.
Then, since $V(\XX)= \{\xx_1,\dots,\xx_n\}$, the vector space 
$\ell^2(\XX)\equiv \ell^2(\XX,m_0) $ is of dimension $n$. 
However, since the scalar product in $\ell^2(\XX)$ is defined with the measure 
$m_0$ while $\C^n$ is endowed with the standard scalar product, one unitary transformation has to be defined. More precisely,
one sets $\I:\ell^2(\XX)\to \C^n$ acting on any $\varphi\in \ell^2(\XX)$ as
\begin{equation}\label{def_de_I}
\I \varphi :=\big(m_0(\xx_1)^{\frac12}\varphi(\xx_1), m_0(\xx_2)^{\frac12}\varphi(\xx_2),\dots,m_0(\xx_n)^{\frac12}\varphi(\xx_n)\big).
\end{equation}
This map defines clearly a unitary transformation between $\ell^2(\XX)$ and $\C^n$.

Let us now consider the Hilbert spaces $\ell^2(X,m_0)$ and $L^2\big(\T^d; \ell^2(\XX)\big)$.
We define the map $\U: C_c(X) \to L^2\big(\T^d; \ell^2(\XX)\big)$ for $f\in C_c(X)$,
$\xi\in \T^d$, and $\xx\in V(\XX)$ by
\begin{equation}\label{def_de_U}
[\U f](\xi,\xx):=\sum_{\mu\in\Gamma}e^{-2\pi i\,\xi\cdot\mu}f(\mu\hat{\xx}).
\end{equation}
Clearly, the map $\U$ corresponds the composition of two maps: the identification of $\ell^2(X,m_0)$ with $\ell^2\big(\Z^d;\ell^2(\XX)\big)$
and the Fourier transform introduced in \eqref{def_Fourier}.
As a consequence, $\U$ extends to a unitary map from $\ell^2(X,m_0)$ to $L^2\big(\T^d;\ell^2(\XX)\big)$, and we shall
keep the same notation for this continuous extension.
The formula for its adjoint is then given on any $\f\in L^1\big(\T^d;\ell^2(\XX)\big)$ by
\begin{equation*}
[\U^*\f](x) = \int_{\T^d}\d\xi\;\!e^{2\pi i  \xi \cdot\lfloor x\rfloor}\f(\xi,\check{x}).
\end{equation*}
Finally, the composed map $\I\U$ provides a unitary map between $\ell^2(X,m_0)$ and $L^2(\T^d;\C^n)$.

We can now state the main result of this section, and refer to \cite[Prop.~4.7]{PR}
for its proof. Note that we use the common notation $\delta_{j\ell}$ for the Kronecker delta function, and that the index map $\eta$ has been introduced in Section 
\ref{basic_topological_crystals}.

\begin{Proposition}\label{prop_def_H}
Let  $(X,\XX,\omega,\Gamma)$ be a topological crystal and let $m_0$ be a $\Gamma$-periodic measure on $X$.
Let $R_0$ be a real $\Gamma$-periodic function defined on $V(X)$.
Then the periodic Schr\"odinger operator $H_0:=-\Delta(X,m_0)+R_0$
is unitarily equivalent to a bounded analytically fibered operator in $L^2(\T^d;\C^n)$, 
namely $\I \U H_0 \U^*\I^*$ is equal to the operator
defined by the function $h_0:\T^d\to M_n(\C)$ with
\begin{equation*}
h_0(\xi)_{j\ell}:=-\sum_{\ee=(\xx_j,\xx_\ell)}\frac{m_0(\ee)}{m_0(\xx_j)^{\frac{1}{2}}\;\!m_0(\xx_\ell)^{\frac{1}{2}}}
\;\! e^{2\pi i\,\xi\cdot\eta(\ee)} + \big(\deg_{m_0}(\xx_j) + R_0(\xx_j)\big)\delta_{j\ell}
\end{equation*}
for any $\xi \in \T^d$ and $j,\ell\in \{1,\dots,n\}$.
\end{Proposition}

\subsection{Mourre theory}\label{subsec_Mourre}

In this section we first recall some definitions related to Mourre theory, such as some
regularity conditions as well as the meaning of a Mourre estimate.
This version of Mourre theory is suitable for bounded operators, 
a more general version and more information are provided in \cite[Chap.~7]{ABG96}.

Let us consider a Hilbert space $\H$ with scalar product
$\langle\;\!\cdot\;\!,\;\!\cdot\;\!\rangle$ and norm $\|\;\!\cdot\;\!\|$. Let
also $S$ and $A$ be two self-adjoint operators in $\H$. The operator $S$ is assumed to be bounded,
and we write $\D(A)$ for the domain of $A$. The spectrum of $S$ is denoted by $\sigma(S)$ and its spectral measure by
$E^S(\;\!\cdot\;\!)$. For shortness, we also use the notation
$E^S(\lambda;\varepsilon):=E^S\big((\lambda-\varepsilon,\lambda+\varepsilon)\big)$
for all $\lambda\in\R$ and $\varepsilon>0$.

The operator $S$ belongs to $C^1(A)$ if the map
\begin{equation}\label{C1}
\R\ni t\mapsto e^{-itA}S e^{itA}\in\B(\H)
\end{equation}
is strongly of class $C^1$ in $\H$. Equivalently, $S\in C^1(A)$ if the quadratic form
\begin{equation*}
\D(A)\ni\varphi
\mapsto\langle iA\varphi,S\varphi\rangle -\langle iS\varphi,A\varphi\rangle \in \C
\end{equation*}
is continuous in the topology of $\H$.
In such a case, this form extends uniquely to a continuous form on $\H$, and the corresponding bounded self-adjoint operator
is denoted by $[iS,A]$.
This $C^1(A)$-regularity of $S$ with respect to $A$ is the basic ingredient for any
investigation in Mourre theory.

Let us also define some stronger regularity conditions. First of all, $S\in C^2(A)$ if the map \eqref{C1}
is strongly of class $C^2$ in $\H$. A weaker condition can be expressed as follows:
$S\in C^{1,1}(A)$ if
\begin{equation*}
\int_0^1  \frac{\d t}{t^2}\;\!\big\| e^{-itA}S e^{itA}+ e^{itA}S e^{-itA}-2S\big\|<\infty.
\end{equation*}
Then, the following inclusions hold: $C^2(A)\subset C^{1,1}(A) \subset C^1(A)$.

For any $S\in C^1(A)$, let us now introduce two subsets of $\R$ which will play a central role. The first one is called the Mourre set of $S$ with respect to $A$, and for the definition of the second one we denote by $\K(\H)$ the set of compact operators on $\H$.
Namely, one sets
\begin{equation*}
\mu^A(S):=\big\{\lambda \in \R\mid \exists \varepsilon>0, a>0 \hbox{ s.t. } E^S(\lambda;\varepsilon)[iS,A]E^S(\lambda;\varepsilon)\geq a E^S(\lambda;\varepsilon) \big\}
\end{equation*}
as well as the larger subset of $\R$ defined by
\begin{align*}
\tilde{\mu}^A(S):=\big\{\lambda \in \R\mid \ & \exists \varepsilon>0, a>0, K\in \K(\H) \hbox{ s.t. }  \\
& \quad E^S(\lambda;\varepsilon)[iS,A]E^S(\lambda;\varepsilon)\geq a E^S(\lambda;\varepsilon)+K \big\}.
\end{align*}

Let us still mention how a perturbative scheme can be developed.
Consider a perturbation $K\in \K(\H)$ and assume that $K$ is self-adjoint and belongs to $C^{1}(A)$.
Even if $\mu^A(S)$ is known, it is usually quite difficult to compute the corresponding set
$\mu^A(S+K)$ for the self-adjoint operator $S+K$.
However, the set $\tilde{\mu}^A(S)$ is much more stable since $\tilde{\mu}^A(S)=\tilde{\mu}^A(S+K)$,
as a direct consequence of \cite[Thm.~7.2.9]{ABG96}.
Based on this observation, the following adaptation of \cite[Thm.~7.4.2]{ABG96} can be stated in our context:

\begin{Theorem}\label{mourreap}
Let $S$ be a self-adjoint element of $\B(\H)$ and assume that $S\in C^{1,1}(A)$.
Let $K\in \K(\H)$ and assume that $K$ is self-adjoint and belongs to $C^{1,1}(A)$.
Then, for any closed interval $I\subset \tilde{\mu}^A(S)$ the operator $S+K$ has at most
a finite number of eigenvalues in $I$, and no singular continuous spectrum in $I$.
\end{Theorem}

In order to use the above framework and results, a conjugate operator for $H_0$
has to be exhibited. The construction of this operator is rather long and has been 
provided with details in \cite{PR}. For that reason, we shall not recall it here, 
but exhibit some important properties which will be used subsequently.
In particular, let us just mention that its construction
does not take place in the initial Hilbert space, but in the space 
$L^2(\T^d;\C^n)$ and for the self-adjoint and bounded analytically fibered 
operator $h_0$ introduced in Proposition \ref{prop_def_H}.
It heavily relies on real analycity theory and on a classical result on stratifications of Hironaka.
More precisely, the so-called Bloch variety
$$
\Sigma:=\big\{(\lambda,\xi)\in\R\times\T^d \mid \lambda\in\sigma\big(h_0(\xi)\big)\big\}
$$
together with its decomposition in a family of semi-analytic sets play a central
role in the construction. In particular, they define  a discrete set $\tau\subset \R$ of 
thresholds on which the spectral analysis of $h_0$ can not be carried out.
By using the local property of the Bloch variety, and by following the seminal ideas
presented in \cite{GN98}, a conjugate operator for $h_0$ is constructed in 
\cite[Sec.~5]{PR} and the main result of that section reads:

\begin{Theorem}[Thm.~5.7 of \cite{PR}]\label{hcinfinity}
Let $h_0$ be the multiplication operator in $L^2(\T^d;\C^n)$ defined by 
the real analytic function $\T^d\to M_n(\C)$ 
introduced in Proposition \ref{prop_def_H}.
Let $\tau$ be the set of thresholds mentioned above and let $I$ be any closed interval in $\R\setminus \tau$.
Then, there exists a self-adjoint operator $A_I$ satisfying the following two properties:
\begin{enumerate}
\item[(i)] the operator $h_0$ belongs to $C^2(A_I)$,
\item[(ii)] there exists a constant $a_I>0$ such that
\begin{equation}\label{eq_Mourre}
E^{h_0}(I)\[ih_0,A_I\]E^{h_0}(I)\ge a_I E^{h_0}(I)\ .
\end{equation}
\end{enumerate}
\end{Theorem}

Let us conclude this section with a few remarks borrowed from \cite{PR}.
First of all, as a consequence of \eqref{eq_Mourre}, it follows that for any 
closed interval $I\equiv [a,b]\subset \R\setminus \tau$, one has
\begin{equation}\label{eq_mu}
(a,b)\subset \mu^{A_I}(h_0)\subset \tilde\mu^{A_I}(h_0).
\end{equation}
Secondly, the operator $A_I$ is essentially
self-adjoint on $C^\infty(\T^d;\C^n)$. Now, let us set
$\Delta_{\T^d}$ for the Laplace operator on $L^2(\T^d)$, and define
$\Lambda:=\big(\Id-\Delta_{\T^d}\big)^{\frac{1}{2}}\otimes \Id_n$ which is a self-adjoint operator in $L^2(\T^d;\C^n)$. This operator satisfies $\D(\Lambda)=\H^1(\T^d;\C^n)$, where $\H^1(\T^d;\C^n)$ is the $1^{\rm{st}}$ Sobolev space on $\T^d$ with values in $\C^n$, and the inclusion $\D(\Lambda)\subset \D(A_I)$ holds.
In addition,  the closure of the operator $\Lambda^{-2} A_I^2$, defined on the domain $\D(A_I^2)$, corresponds to a bounded operator in $L^2(\T^d;\C^n)$.
This information will be used later for an application of the abstract result for short-range type perturbations presented in \cite[Thm.~7.5.8]{ABG96}.

\section{Proof of the main result}\label{sec_pert}
\setcounter{equation}{0}

In this section we provide the proof of the main result.
At a technical level, our work consists in considering the difference between the 
operator $H$ introduced in \eqref{full}
and the periodic operator $H_0$ introduced in 
\eqref{h0vertices}, and to show that this difference 
belongs to $C^{1,1}(A_{I})$. 
An application of Mourre theory will then lead to the results.

The first result is obtained by a simple computation, using the unitary transformations $\J$, 
$\I$, and $\U$ introduced respectively in \eqref{eq_map_J}, \eqref{def_de_I}, and \eqref{def_de_U}. 
For its statement, let us define the following convenient map:
\begin{equation*}
\imath:V(X)\to\{1,\dots,n\}, \qquad x_{\imath(x)}:=\widehat{\check{x}},
\end{equation*}
which associates to any $x\in V(X)$ the index of the representative $x_j\in V(X)$
which belongs to the same orbit under the action of the group $\Z^d$.
Note that we shall also use the natural identification of $\F$ with $\F \otimes \Id_{n}$
whenever necessary.

\begin{Lemma}\label{lem_dif}
In the framework considered above, 
the following equality holds in $L^2(\T^d;\C^n)$\;\!:
\begin{equation}\label{eq_decomposition}
\I \U \big( \Delta(X,m_{0}) - \J^* \Delta(\X, m) \J \big) \U^{*} \I^{*} 
=  \Op(b) + \F L_-\F^* -\F L_+\F^*,
\end{equation}
where $\Op(b)$ is the toroidal pseudodifferential operator with symbol $b$ defined in \cite[Prop.~6.6]{PR} and with $m(\e):=m_0(\e)$ for any $\e \in F_-$, and $L_\pm$ are given on $\varphi \in C_c(\Z^d;\C^n)$ and $\mu \in \Z^d$ by 
\begin{equation}\label{eq_K_-}
[L_- \varphi]_{j}(\mu) 
:=  
\sum_{\e \in A(F_-)_{\mu x_{j}}} \frac{m_0(\e)}{m(\mu x_{j})^{1/2} m(t(\e))^{1/2}} \varphi_{\iota(t(\e))}(\lfloor t(\e) \rfloor) - 
\sum_{\e \in A(F_-)_{\mu x_{j}}} \frac{m_0(\e)}{m(\mu x_{j})} \varphi_{j}(\mu) 
\end{equation}
and 
\begin{equation}\label{eq_K_+}
[L_+ \varphi]_{j}(\mu)
:= 
\sum_{\e \in A(F_+)_{\mu x_{j}}} \frac{m(\e)}{m(\mu x_{j})^{1/2} m(t(\e))^{1/2}} \varphi_{\iota(t(\e))}(\lfloor t(\e) \rfloor) - 
\sum_{\e \in A(F_+)_{\mu x_{j}}} \frac{m(\e)}{m(\mu x_{j})} \varphi_{j}(\mu). 
\end{equation}
\end{Lemma}

In the previous statement, the term $\Op(b)$ has been thoroughly studied in \cite{PR}, 
we shall not reproduce its analysis here. On the other hand, we shall concentrate
on the terms $L_\pm$.
The summations in $L_-$ contain only a finite number of contributions, and this term can be easily treated separately. Alternatively, since the two expressions \eqref{eq_K_-} and \eqref{eq_K_+}  are formally the same, the study of $L_-$ can be mimicked from the analysis of $L_+$ provided below, once it is observed that the conditions of the next proposition are always satisfied for the finite set $F_-$.

For the next statement, recall that the partial degree function $\deg_{F_+}$
has been introduced in \eqref{def_deg_F}.

\begin{Proposition}\label{decay_of_F}
Assume that the measure $m$ satisfies  
\begin{equation}\label{decay_of_F_1}
\int_{1}^{\infty} \d \lambda \sup_{\lambda < | \lfloor x \rfloor| < 2\lambda} 
\deg_{F_+}(x) < \infty,
\end{equation}
\begin{equation}\label{decay_of_F_2}
\int_{1}^{\infty} \d \lambda \sup_{ x\in V(\X) } 
\sqrt{\sum_{\substack{\e \in A(F_+)_{x} \\ \lambda\leq | \lfloor t(\e) \rfloor| \leq2 \lambda}}  \frac{m(\e)}{m(o(\e))}}  < \infty.
\end{equation}
Then the term  $\F L_+\F^*$ belongs to $C^{1,1}(A_{I})$. 
\end{Proposition}

\begin{proof}
This proof consists in an application of an abstract result presented in \cite[Thm.~7.5.8]{ABG96}. 
We shall thus check the assumptions of this theorem with $\GG=\HH=L^{2}(\T^{d};\C^{n})$ and $\Lambda:=\big(\Id-\Delta_{\T^d}\big)^{\frac{1}{2}}\otimes \Id_n$. 
Thanks to the information at the end of section \ref{subsec_Mourre}, 
it suffices to show that there exists $\theta \in C_{c}^{\infty}((0,\infty))$ not identically zero such that 
\begin{equation}\label{eq_condition}
\int_{1}^{\infty}\d\lambda \left\| \theta\left( \frac{\Lambda}{\lambda} \right)\F L_+ \F^* \right\|_{\B(L^{2}(\T^{d};\C^{n}))} < \infty.
\end{equation}
 
From now we consider $\theta \in C_{c}^{\infty}\big((0,\infty);[0,1]\big)$ with support 
contained in $(\sqrt{2}, 2)$.
Then one has
\begin{equation*}
\left\| \theta\left( \frac{\Lambda}{\lambda} \right) \F L_+ \F^*\right\|_{\B(L^{2}(\T^{d};\C^{n}))} =
\left\| \theta\left( \frac{\langle N \rangle}{\lambda} \right)   L_+ \right\|_{\B(\ell^{2}(\Z^{d};\C^{n}))},
\end{equation*}
where $\langle N\rangle$ denotes the multiplication operator by the function
$\mu\mapsto (1+|\mu|^2)^\frac{1}{2}$ in $\ell^{2}(\Z^{d};\C^{n})$.
For $\varphi=(\varphi_{1}, \ldots, \varphi_{n})$ with each $\varphi_{j} \in C_c(\Z^{d})$, one has
\begin{align}
\nonumber & \left\| \theta\left( \frac{\langle N \rangle}{\lambda} \right)  L_+ 
\varphi \right\|^{2}_{\ell^{2}(\Z^{d};\C^{n})} \\
\nonumber &= \sum_{\mu\in\Z^{d}}\sum_{j=1}^{n} \bigg| \theta\left(\frac{\langle \mu \rangle}{\lambda}\right) 
\sum_{\e \in A(F_+)_{\mu x_{j}}} \Big( \frac{m(\e)}{m(\mu x_{j})^{1/2} m(t(\e))^{1/2}} \varphi_{\iota(t(\e))}(\lfloor t(\e) \rfloor) 
- \frac{m(\e)}{m(\mu x_{j})} \varphi_{j}(\mu) \Big) \bigg|^{2} \\
\nonumber &
\leq 2 \sum_{\mu\in\Z^{d}}\sum_{j=1}^{n}  \theta\left(\frac{\langle \mu \rangle}{\lambda}\right)^{2}  \bigg( \Big| \sum_{\e \in A(F_+)_{\mu x_{j}}} 
\frac{m(\e)}{m(\mu x_{j})^{1/2} m(t(\e))^{1/2}} \varphi_{\iota(t(\e))}(\lfloor t(\e) \rfloor) \Big|^{2} \\
\label{eq_second_t} &\qquad\qquad \qquad \qquad \qquad \quad + \Big| \sum_{\e \in A(F_+)_{\mu x_{j}}}  \frac{m(\e)}{m(\mu x_{j})} \varphi_{j}(\mu) \Big|^{2} \bigg) .
\end{align}

Using Fubini's theorem and considering opposite arrows $\bar{\e}$ instead of $\e$, 
the first term can be estimated as follows:
\begin{align*}
& \sum_{\mu\in\Z^{d}}\sum_{j=1}^{n} \bigg| \theta\left(\frac{\langle \mu \rangle}{\lambda}\right) \sum_{e \in A(F_+)_{\mu x_{j}}} 
\frac{m(\e)}{m(\mu x_{j})^{1/2} m(t(\e))^{1/2}} \varphi_{\iota(t(\e))}(\lfloor t(\e) \rfloor) \bigg|^{2} \\
&\leq \sum_{\mu\in\Z^{d}}\sum_{j=1}^{n} \theta\left(\frac{\langle \mu \rangle}{\lambda}\right)^2 \bigg(\sum_{\e \in A(F_+)_{\mu x_{j}}} \frac{m(\e)}{m(\mu x_j)}\bigg)
\bigg(\sum_{\e \in A(F_+)_{\mu x_{j}}} \frac{m(\e)}{m(t(\e))} \big|\varphi_{\iota(t(\e))}(\lfloor t(\e) \rfloor)\big|^2\bigg)  \\
&\leq C \sum_{\mu\in\Z^{d}}\sum_{j=1}^{n} \theta\left(\frac{\langle \mu \rangle}{\lambda}\right)^2  \sum_{\e \in A(F_+)_{\mu x_{j}}} \frac{m(\e)}{m(t(\e))}
\big|\varphi_{\iota(t(\e))}(\lfloor t(\e) \rfloor)\big|^2 \\
& = C \sum_{\e\in A(F_+)} \theta\left(\frac{\langle \lfloor o(\e) \rfloor \rangle}{\lambda}\right)^2 \frac{m(\e)}{m(t(\e))} \big|\varphi_{\iota(t(\e))}(\lfloor t(\e) \rfloor) \big|^2 \\
&= C \sum_{\e\in A(F_+)} \theta\left(\frac{\langle \lfloor t(\e) \rfloor \rangle}{\lambda}\right)^2 \frac{m(\e)}{m(o(\e))} \big|\varphi_{\iota(o(\e))}(\lfloor o(\e) \rfloor) \big|^{2} \\
&= C \sum_{\mu\in\Z^{d}}\sum_{j=1}^{n}
\sum_{\e \in A(F_+)_{\mu x_{j}}}  \theta\left( \frac{\langle \lfloor t(\e) \rfloor \rangle}{\lambda} \right)^2 \frac{m(\e)}{m(\mu x_j)} \big|\varphi_{j}(\mu) \big|^{2} \\
&\leq C \sum_{\mu\in\Z^{d}}\sum_{j=1}^{n} 
\sum_{\substack{\e \in A(F_+)_{\mu x_{j}} \\ 2\lambda^{2}-1 \leq | \lfloor t(\e) \rfloor|^2 \leq4\lambda^{2}-1 }} \frac{m(\e)}{m(\mu x_j)} \big|\varphi_{j}(\mu) \big|^{2},
\end{align*}
where $C$ is a bound for the function $\deg_m$.
Thus, if we define $\vartheta_1: \Z^{d}\times [1,\infty)\to M_{n}(\C)$ by 
\begin{equation*}
\vartheta_1(\mu,\lambda)_{jj} :=
\sqrt{\sum_{\substack{\e \in A(F_+)_{\mu x_{j}} \\ 2\lambda^{2}-1 \leq | \lfloor t(\e) \rfloor|^2 \leq4\lambda^{2}-1 }} \frac{m(\e)}{m(\mu x_j)}} 
\end{equation*}
and $\vartheta_1(\mu,\lambda)_{j\ell}=0$ if $j\neq \ell$, then we obtain
\begin{align*}
\sum_{\mu\in\Z^{d}}\sum_{j=1}^{n} 
\sum_{\substack{\e \in A(F_+)_{\mu x_{j}} \\ 2\lambda^{2}-1 \leq | \lfloor t(\e) \rfloor|^2 \leq4\lambda^{2}-1 }} \frac{m(\e)}{m(\mu x_j)} \big|\varphi_{j}(\mu) \big|^{2}
&= \big\|\vartheta_1(N,\lambda)\varphi\big\|_{\ell^2(\Z^d;\C^n)}^2 \\
&\leq \big\|\vartheta_1(N,\lambda)\big\|_{\B(\ell^2(\Z^d;\C^n))}^2
\|\varphi\|_{\ell^2(\Z^d;\C^n)}^2.
\end{align*}
Here the notation $\vartheta_1(N,\lambda)$ means simply the multiplication operator
by the function $\mu\mapsto \vartheta_1(\mu,\lambda)$.
Since $\lambda^2\leq 2\lambda^2-1$ for any $\lambda\geq 1$, and since
$4\lambda^2-1<4\lambda^2$, 
we finally observe that
\begin{align}
\nonumber \big\|\vartheta_1(N,\lambda)\big\|_{\B(\ell^2(\Z^d;\C^n))}
&=\sup_{\mu\in\Z^{d}}\max_{1\leq j\leq n}\vartheta_1(\mu,\lambda)_{jj} \\
\nonumber &=\sup_{\mu\in\Z^{d}}\max_{1\leq j\leq n}
\sqrt{\sum_{\substack{\e \in A(F_+)_{\mu x_{j}} \\ 2\lambda^{2}-1 \leq | \lfloor t(\e) \rfloor|^2 \leq4\lambda^{2}-1 }}  \frac{m(\e)}{m(\mu x_j)}} \\
\nonumber & \leq  \sup_{\mu\in\Z^{d}}\max_{1\leq j\leq n}
\sqrt{\sum_{\substack{\e \in A(F_+)_{\mu x_{j}} \\ \lambda\leq | \lfloor t(\e) \rfloor| \leq2 \lambda}}  \frac{m(\e)}{m(\mu x_j)}} \\
\label{eq_1t} & =  \sup_{x\in V(\X)}
\sqrt{\sum_{\substack{\e \in A(F_+)_{x} \\ \lambda\leq | \lfloor t(\e) \rfloor| \leq2 \lambda}}  \frac{m(\e)}{m(o(\e))}}. 
\end{align}

For the second term in \eqref{eq_second_t}, let us define $\vartheta_2:\Z^{d}\to M_{n}(\C)$ by 
$$
\vartheta_2(\mu)_{jj} :=
\sum_{\e \in A(F_+)_{\mu x_{j}}} \frac{m(\e)}{m(\mu x_{j})}
$$ 
and $\vartheta_2(\mu)_{j\ell}=0$ if $j\neq \ell$.
Then the second term in \eqref{eq_second_t} can be computed as
\begin{equation*}
\sum_{\mu\in\Z^{d}}\sum_{j=1}^{n} \bigg| \theta\left(\frac{\langle \mu \rangle}{\lambda}\right)
\sum_{\e \in A(F_+)_{\mu x_{j}}} \frac{m(\e)}{m(\mu x_{j})} \varphi_{j}(\mu) \bigg|^{2} 
= \left\| \theta\left( \frac{\langle N \rangle}{\lambda} \right)  \vartheta_2(N)\varphi \right\|^{2}_{\ell^2(\Z^d;\C^n)},
\end{equation*}
and one has for $\lambda\geq 1$
\begin{equation}\label{eq_2t}
\left\| \theta\left( \frac{\langle N \rangle}{\lambda} \right)  \vartheta_2(N) \right\|_{\B(\ell^2(\Z^d;\C^n))} \leq
\sup_{\lambda< | \mu | < 2\lambda} \max_{1\leq j \leq n} \sum_{\e\in A(F_+)_{\mu x_{j}}} \frac{m(\e)}{m(\mu x_j)}. 
\end{equation}
By inserting  the estimates \eqref{eq_1t} and \eqref{eq_2t} in \eqref{eq_condition},
one obtains the assumptions of the statement.
As a consequence, one has checked all assumptions of \cite[Thm.~7.5.8]{ABG96}, from which one deduces that $\F L_+ \F^*$ belongs to $C^{1,1}(A_{I})$. 
\end{proof}

In the next statement, we summarize the regularity result. 
The initial framework is a topological crystal
$(X,\XX,\omega,\Z^d)$  with $X=\big(V(X),E(X)\big)$ 
together with a $\Z^d$-periodic measure $m_0$ 
and a $\Z^d$-periodic function $R_0:V(X)\to \R$. 

\begin{Proposition}\label{prop_full_reg}
Let $F_+$ be a possibly infinite set of unoriented new edges, let $F_{-}$ be a finite subset of $E(X)$, and consider the graph $\X=\big(V(\X),E(\X)\big)$ given by
$V(\X):=V(X)$ and $E(\X):= \big(E(X)\setminus F_-\big)\cup F_+$.
Assume that the measure $m$ on $\X$ satisfies 
\begin{equation}\label{eq_decay_m}
\int_{1}^{\infty}\d\lambda \sup_{\substack{\e\in E(X)\setminus F_- \\ \lambda<|\lfloor \e\rfloor |<2\lambda}}\left|\frac{m(\e)}{m(o(\e))}-\frac{m_0(\e)}{m_0(o(\e))}\right|<\infty,
\end{equation}
together with conditions \eqref{decay_of_F_1} and \eqref{decay_of_F_2}.
Assume also that the function $R:V(\X)\to \R$ satisfies the decay condition
\begin{equation}\label{eq_decay_R}
\int_{1}^{\infty}\d\lambda \sup_{\lambda<|\lfloor x\rfloor |<2\lambda}\left|R(x)-R_0(x)\right|<\infty,
\end{equation}
Then the difference  $\I \U \Big(\big(- \Delta(X,m_{0}) +R_0\big)- \J^* \big(-\Delta(\X, m)+R\big) \J \Big) \U^{*} \I^{*}$ belongs to $C^{1,1}(A_{I})$. 
\end{Proposition}

\begin{proof}
The proof simply consists in observing that Proposition \ref{decay_of_F} together
with some results of \cite{PR} imply the statement. Indeed, instead of setting $m(\e)=0$
for any $\e \in F_-$ let us set $m(\e):=m_0(\e)$ for any $\e\in F_-$. Then, since
$m(\e)>0$ for all $\e\in E(X)$, and since the assumptions \eqref{eq_decay_m}
and \eqref{eq_decay_R} imply the assumptions of Lemmas 6.7 and 6.8 of \cite{PR},
it follows that the term $\Op(b)$ mentioned in Lemma \ref{lem_dif}
and the difference $\I\U(R_0-\J^* R \J)\U^*\I^*$ belong to $C^{1,1}(A_I)$.
For the remaining two contributions exhibited in Lemma \ref{lem_dif}, the term
$\F L_+\F^*$ has been treated in details in Proposition \ref{decay_of_F}.
The term $\F L_-\F^*$, which is much simpler because it contains only a finite sum,
can be treated as $\F L_+\F^* $, once it is observed that conditions 
\eqref{decay_of_F_1} and \eqref{decay_of_F_2} are satisfied for $F_-$ replacing $F_+$.
\end{proof}

\begin{proof}[Proof of Theorem \ref{thm_main}]
For the first two statements of the main theorem, we shall rely on Theorem \ref{mourreap}
with $S=h_0$ and $K$ defined by 
\begin{equation}\label{eq_difference}
\I\U\big(H_0 -\J^* H\J\big)\U^*\I^* = 
-\Op(b) - \F L_-\F^*  + \F L_+\F^* +\I\U(R_0-\J^* R \J)\U^*\I^*,
\end{equation}
where we have used the expressions provided in \eqref{eq_decomposition}.
Clearly, the regularity condition on $h_0$ follows from Theorem \ref{hcinfinity}.
Also, as a consequence of Proposition \ref{prop_full_reg} the terms in 
\eqref{eq_difference} belong to $C^{1,1}(A_I)$.
Let us now check the compactness of the terms defined in the r.h.s.~of \eqref{eq_difference}.

The compactness of the operator $\Op(b)$ has already been obtained
in the proof of \cite[Thm.~2.3]{PR}, under the assumption \eqref{eq_Cm1}.
Since $L_{-}$ is a finite rank operator, it is compact.
For the proof of the compactness of $L_{+}$ and for each $m\in\N$, let $L_{+,m}$ be the finite rank operator defined by $L_{+,m}:= \chi_{[0,m]}(\langle N\rangle)L_{+}$ where $\chi_{[0,m]}$ is the characteristic function on the interval $[0,m]$.
Consider also $\eta\in C^\infty\big((0,\infty);[0,1]\big)$ satisfying 
\begin{equation*}
\eta (s) := 
\begin{cases}
0 & \quad \text{if} \ s\leq\sqrt{2} \\
1 & \quad \text{if} \ s\geq 2.   
\end{cases}
\end{equation*}
According to \cite[Rem.~7.6.9]{ABG96}, the decay
assumptions \eqref{eq_Cm2} and \eqref{eq_Cm3} also
imply the estimate $\int_{1}^{\infty} \d \lambda \big\| \eta\big( \frac{\langle N\rangle}{\lambda}\big) L_{+} \big\|_{\B(\ell^{2}(\Z^{d};\C^{n}))} < \infty$, 
as in the proof of Proposition \ref{decay_of_F}.
In addition, since the expression $\big\| \eta\big( \frac{\langle N\rangle}{\lambda}\big) L_{+} \big\|_{\B(\ell^{2}(\Z^{d};\C^{n}))}$ is non-increasing as $\lambda$ increases, it readily follows that  
$\big\| \eta\big( \frac{\langle N\rangle}{\lambda}\big) L_{+} \big\|_{\B(\ell^{2}(\Z^{d};\C^{n}))} \to 0$ as $\lambda \to \infty$.
Recall now that the operator $\Lambda:=\big(\Id-\Delta_{\T^d}\big)^{\frac{1}{2}}\otimes \Id_n$ has been introduced at the end of Section \ref{subsec_Mourre},
From the above argument, one infers that  
\begin{align*}
\lim_{m\to \infty}\|\F L_+\F^* -\chi_{[0,m]}(\Lambda)\F L_+\F^*\|_{\B(L^2(\T^d;\C^n))} 
& =\lim_{m\to \infty} \| L_{+} - L_{+,m} \|_{\B(\ell^{2}(\Z^{d};\C^{n}))} \\
& = \lim_{m\to \infty} \| \chi_{(m,\infty)}(\langle N\rangle)L_{+} \|_{\B(\ell^{2}(\Z^{d};\C^{n}))} \\
& = 0
\end{align*}
which leads to the compactness of $\F L_+ \F^*$. 
Note also that the compactness of $R_0-\J^* R \J$ follows by a similar 
but much simpler argument, using the decay condition \eqref{eq_Cm4}.
 
We are thus in a suitable position for using Theorem \ref{mourreap}
with $S=h_0$ and $K$ defined by \eqref{eq_difference}. 
For $\tilde{\mu}^{A_I}(h_0)$, one can use the result obtained in \eqref{eq_mu},
by considering a slightly bigger interval $I'$ with $I\subset I'\subset \R\setminus \tau$.
Then, statements 1. and 2. of Theorem \ref{thm_main} follow from Theorem \ref{mourreap} by taking into account the conjugation by the unitary
transform $\I\U$.

For the existence and completeness of the wave operators, observe first that since $\J$ is unitary, these properties for $W_\pm(H,H_0;\J,I)$ are equivalent to
the same properties for $W_\pm(\J^* H\J,H_0;I)$. Then, by using again the unitary transform $\I\U$, one observes that this is still equivalent to
the existence and the completeness of
\begin{equation}\label{eq_local}
W_\pm(\I\U \J^* H\J \U^*\I^*,\I\U H_0 \U^*\I^*;I).
\end{equation}
Such properties will now be deduced from \cite[Theorem 7.4.3]{ABG96}. Indeed, according to
that statement, if the difference \eqref{eq_difference} belongs to $\B(\KK^{*\circ},\KK)$,
with $\KK:=\big(\D(A_I),L^2(\T^d;\C^n)\big)_{\frac{1}{2},1}$ and $\KK^{*\circ}$ the closure of $L^2(\T^d;\C^n)$ in $\KK^*$,
then the local wave operators \eqref{eq_local} exist and are complete.

In order to check this condition, recall that the operator $\Lambda$ satisfies $\D(\Lambda)\subset \D(A_I)$.
It then follows that $\LL:=\big(\D(\Lambda),L^2(\T^d;\C^n)\big)_{\frac{1}{2},1}\subset \big(\D(A_I),L^2(\T^d;\C^n)\big)_{\frac{1}{2},1}$, \
as shown for example in \cite[Corol.~2.6.3]{ABG96},
and then $\B(\LL^{*\circ},\LL)\subset \B(\KK^{*\circ},\KK)$.
However, we shall still consider the Fourier transform version of these spaces.
More precisely, let us set $\NNN:=\F^* \big(\D(\Lambda),L^2(\T^d;\C^n)\big)_{\frac{1}{2},1}$
which is equal to $\big(\D(\langle N\rangle),l^2(\Z^d;\C^n)\big)_{\frac{1}{2},1}$. Accordingly, one has to show that
\begin{equation}\label{eq_inclusion}
Z:=\F^*\I\U\big(\J^* H\J-H_0\big)\U^*\I^*\F \ \in \B(\NNN^{*\circ},\NNN).
\end{equation}

Fortunately, this term has already been computed (here or in \cite[Sec.~6]{PR}), and when acting on $\varphi=(\varphi_{1}, \ldots, \varphi_{n})$ with each $\varphi_{j} \in C_c(\Z^{d})$,
it is given by
\begin{equation}\label{nearly_the_end}
[Z\varphi]_j(\mu)
= \Big(\Big[\sum_{\ee\in A(\XX)}\Big([T(\ee)](N)\varphi- S_{\eta(\ee)}[K(\ee)](N)\Big) + r_s(N)-L_++L_-\Big]\varphi\Big)_j(\mu).
\end{equation}
Here the operators $[K(\ee)](N)$ and $[T(\ee)](N)$ are matrix valued multiplication operators defined by the functions
$K(\ee):\Z^d\to M_n(\C)$ and $T(\ee):\Z^d\to M_n(\C)$ with
\begin{equation*}
\[K(\ee)\](\mu)_{j\ell}:=\left\{\begin{matrix} \(\frac{m((\mu-\eta(\ee))\hat{\ee})}{m((\mu-\eta(\ee)) o(\hat{\ee}))^{\frac12}m((\mu-\eta(\ee)) t(\hat{\ee}))^{\frac12}}-\frac{m_0(\ee)}{m_0(o(\ee))^{\frac12}m_0(t(\ee))^{\frac12}}\) & \hbox{ if } o(\ee) = \xx_j,  t(\ee)=\xx_\ell \\
0 &  \hbox{otherwise}\end{matrix}\right.
\end{equation*}
and
\begin{equation*}
\[T(\ee)\](\mu)_{j \ell}:=\left\{\begin{matrix} \(\frac{m(\mu\hat{\ee})}{m(\mu o(\hat{\ee}))}-\frac{m_0(\ee)}{m_0(o(\ee))}\) & \hbox{ if } o(\ee)=\xx_j \hbox{ and } j=\ell \\
0 &  \hbox{otherwise} \end{matrix} \right.
\end{equation*}
with the convention that $m(\e):=m_0(\e)$ for any $\e \in F_-$.
Also, the multiplication operator $r_s(N)$ is defined by the function 
$r_s:\Z^d\to M_n(\C)$ with
\begin{equation*}
r_s(\mu)_{j\ell}:=\big(R(\mu x_j)-R_0(\mu x_j)\big)\delta_{j\ell}.
\end{equation*}
For any $\nu$ we have also used $S_\nu$ for the shift operator by 
$\nu$ acting on any $\varphi \in l^2(\Z^d;\C^n)$ as $[S_\nu \varphi](\mu):=\varphi(\mu+\nu)$.

In \cite[Sec.~6]{PR} it has been shown that an estimate of the form
\begin{equation}\label{condM}
\int_1^\infty \d \lambda  \sup_{\lambda<|\mu|<2\lambda }\lp M(\mu)\rp_{\B(\ell^2(\Z^d;\C^n))} <\infty,
\end{equation}
is satisfied for $M$ replaced by $K(\ee)$, $T(\ee)$ or by $r_s$.
Then, by an application of \cite[Lem.~6.3]{PR} one immediately deduces that
the operators  $K(\ee)$, $S_{\eta(\ee)}T(\ee)$, and $r_s$
belong to $\B(\NNN^{*\circ},\NNN)$, and so does the finite sum of them
appearing in \eqref{nearly_the_end}.
Note that the same argument applies for the second term appearing in the definition of $L_+$, see \eqref{eq_K_+}.
Indeed, by setting $M(\mu)_{jj}:=\deg_{F_+}(\mu x_j)$ and $M(\mu)_{j\ell}=0$ if $j\neq \ell$, then the assumption
\eqref{eq_Cm2} implies that the above condition \eqref{condM} is satisfied, 
and therefore \cite[Lem.~6.3]{PR} can also be applied.
Obviously, the same is true for the second term of $L_-$, since the sum is finite.

Unfortunately, the same approach does not hold for the first term in $L_+$
since this term is highly non-local. 
Note that the same is true for the first term in $L_-$, but since the sum in this term is finite, 
it can be easily treated (or treated like $L_+$).
For the first term $L_{+}^1$ in $L_+$ we shall impose a slightly stronger condition, 
namely we shall impose that 
$L_{+}^1\in \B(\GG^*,\GG)$, with $\GG:=\D(\langle N\rangle^s)$ for some $s>1/2$.
Then, since $\D(\langle N\rangle^s)\subset \big(\D(\langle N\rangle),l^2(\Z^d;\C^n)\big)_{\frac{1}{2},1}=:\NNN$, as shown in \cite[Prop.~2.4.1 \& 2.8.1]{ABG96} it follows that $L_{+}^1\in  \B(\NNN^{*\circ},\NNN)$, as required.
Thus, we are left with proving that $\langle N\rangle^s L_{+}^1 \langle N\rangle^s\in \B\big(\ell^2(\Z^d;\C^n)\big)$ for some $s>1/2$.

As in the proof of Proposition \ref{decay_of_F} one has
\begin{align}
\nonumber & \left\|\langle N\rangle^s L_{+}^1\langle N\rangle^s \varphi \right\|^{2}_{\ell^{2}(\Z^{d};\C^{n})} \\
\nonumber &= \sum_{\mu\in\Z^{d}}\sum_{j=1}^{n} \Big| \langle \mu\rangle^s 
\sum_{\e \in A(F_+)_{\mu x_{j}}} \frac{m(\e)}{m(\mu x_{j})^{1/2} m(t(\e))^{1/2}} \langle \lfloor t(\e) \rfloor \rangle^s \varphi_{\iota(t(\e))}(\lfloor t(\e) \rfloor)  \Big|^{2} \\
\nonumber &\leq C \sum_{\mu\in\Z^{d}}\sum_{j=1}^{n}  
\sum_{\e \in A(F_+)_{\mu x_{j}}} \langle \mu\rangle^{2s} \frac{m(\e)}{m(t(\e))} \langle \lfloor t(\e) \rfloor \rangle^{2s} \big|\varphi_{\iota(t(\e))}(\lfloor t(\e) \rfloor)\big|^2  \\
\nonumber &= C \sum_{\mu\in\Z^{d}}\sum_{j=1}^{n}  
\sum_{\e \in A(F_+)_{\mu x_{j}}} \langle \lfloor t(\e) \rfloor  \rangle^{2s} \frac{m(\e)}{m(\mu x_j)} \langle \mu \rangle^{2s} \big|\varphi_{j}(\lfloor \mu \rfloor)\big|^2,
\end{align}
where $C$ is a bound for the function $\deg_m$.
Thus, if we define $\vartheta_3: \Z^{d}\to M_{n}(\C)$ by 
\begin{equation*}
\vartheta_3(\mu)_{jj} :=
\sqrt{ \langle \mu \rangle^{2s} \sum_{\e \in A(F_+)_{\mu x_{j}}} \frac{m(\e)}{m(\mu x_j)} \langle \lfloor t(\e) \rfloor  \rangle^{2s}} 
\end{equation*}
and $\vartheta_3(\mu)_{j\ell}=0$ if $j\neq \ell$, then we obtain
\begin{equation*}
\left\|\langle N\rangle^s L_{+}^1\langle N\rangle^s \varphi \right\|^{2}_{\ell^{2}(\Z^{d};\C^{n})} 
\leq C \big\|\vartheta_3(N)\big\|_{\B(\ell^2(\Z^d;\C^n))}^2
\|\varphi\|_{\ell^2(\Z^d;\C^n)}^2.
\end{equation*}
Since \eqref{eq_Cm5} is precisely the condition about the boundedness of $\vartheta_3$, one directly
infers that $\langle N\rangle^s L_{+}^1 \langle N\rangle^s\in \B\big(\ell^2(\Z^d;\C^n)\big)$ for some $s>1/2$,
and as a consequence the inclusion in \eqref{eq_inclusion} holds.
\end{proof}

Let us now prove that the examples mentioned in Section \ref{sec_intro} satisfy
the conditions of Theorem \ref{thm_main}.
For some computations we shall use the norm $|\cdot|_{\infty}$ on $\Z^{d}$ defined by $|\mu|_{\infty} =\max_{1\leq j \leq d}|\mu_{j}|$.
The inequalities $|\mu|_{\infty}\leq |\mu|\leq \sqrt{d}|\mu|_{\infty}$ 
will also be used at several places.

\begin{proof}[Proof of Example \ref{ex_0_connect}]
Let us first check that $\deg_m$ is bounded. Indeed one has for any $x\in \Z^d$
\begin{equation*}
\deg_m(x) =
\sum_{\e\not \in A(F)_x} m(\e) +\sum_{\e\in A(F)_x}m(\e) 
\leq 2d+C \begin{cases} \sum_{y\neq 0} |y|^{\alpha} & \quad \text{if} \ x=0 \\
|x|^{\alpha} & \quad \text{if} \ x\neq 0.   \end{cases}
\end{equation*}
Let $\Z_{>0}$ be the set of positive integers, and define $S_{r}:=\{ y\in \Z^{d} \mid |y|_{\infty}=r \}$ for $r\in\Z_{>0}$. 
If we also define $B_{r}:=\{ y\in \Z^{d} \mid |y|_{\infty}\leq r \}$, 
then the cardinality $|S_{r}|$  of $S_r$ is estimated as 
\begin{equation*}
|S_{r}| = |B_{r}|-|B_{r-1}| 
= (2r+1)^{d} - (2r-1)^d 
\leq M r^{d-1}
\end{equation*}
for some $M$ independent of $r$. 
Taking this inequality into account, one infers that
\begin{equation*}
\sum_{y\neq 0} |y|^{\alpha} = \sum_{r\in \Z_{>0}}\sum_{y\in S_{r}} |y|^{\alpha}  
\leq \sum_{r\in \Z_{>0}}\sum_{y\in S_{r}} |y|_{\infty}^{\alpha} 
\leq M \sum_{r \in \Z_{>0}}r^{\alpha+d-1}
\end{equation*}
which is clearly bounded if $\alpha < -d$. 

Let us now check that conditions \eqref{eq_Cm2} and \eqref{eq_Cm3} are also satisfied.
Indeed, by considering $|x|\geq \lambda\geq 1$, we obtain for \eqref{eq_Cm2} 
\begin{equation*}
\int_{1}^{\infty} \d\lambda \sup_{\lambda < |x| < 2\lambda}\sum_{\e\in A(F)_{x}}m(\e)
\leq C \int_{1}^{\infty}  \lambda^{\alpha} \d\lambda, 
\end{equation*} 
which is finite for $\alpha < -1$. 
For \eqref{eq_Cm3}, observe that 
\begin{align*}
\sum_{\lambda\leq |y|\leq 2\lambda} |y|^{\alpha}
& \leq \sum_{\frac{\lambda}{\sqrt{d}} \leq |y|_{\infty} \leq 2\lambda} |y|_{\infty}^{\alpha} \\
&\leq M\sum_{\frac{\lambda}{\sqrt{d}} \leq r \leq 2\lambda} r^{\alpha+d-1} \\
&\leq M\Big( \int_{\frac{\lambda}{\sqrt{d}}}^{2\lambda} r^{\alpha+d-1} \d r + (\frac{\lambda}{\sqrt{d}})^{\alpha+d-1} \Big) \\
&= M \Big( \frac{1}{\alpha + d}(2\lambda)^{\alpha+d} - \frac{1}{\alpha + d}(\frac{\lambda}{\sqrt{d}})^{\alpha+d} + (\frac{\lambda}{\sqrt{d}})^{\alpha+d-1} \Big). 
\end{align*}
Then, we infer that
\begin{align*}
\int_{1}^{\infty} \d\lambda \sup_{x\in\Z}
\sqrt{\sum_{\substack{\e\in A(F)_{x} \\ \lambda \leq \lfloor |t(\e)| \rfloor \leq 2\lambda}}m(\e)}
&\leq \int_{1}^{\infty} \d\lambda \sqrt{C \sum_{\lambda\leq |y|\leq 2\lambda} |y|^{\alpha}} \\
&\leq \int_{1}^{\infty} \d\lambda \sqrt{ CM \Big( \frac{1}{\alpha + d}(2\lambda)^{\alpha+d} - \frac{1}{\alpha + d}(\frac{\lambda}{\sqrt{d}})^{\alpha+d} + (\frac{\lambda}{\sqrt{d}})^{\alpha+d-1} \Big)}. 
\end{align*}
The last term is finite if $\frac{\alpha+d}{2} < -1$, that is if $\alpha < -d-2$. 

Let us check condition \eqref{eq_Cm5}. 
We consider $s=\frac{3}{4}$, then we can easily show that condition \eqref{eq_Cm5} is satisfied if $\alpha<-d-\frac{3}{2}$. 
Therefore, under $\alpha<-d-2$, the local wave operators exist and are complete. 
\end{proof}

\begin{proof}[Proof of Example \ref{ex_all_connect}]
Let us first check that $\deg_m$ is bounded. Indeed one has for any $x\in \Z^{d}$
\begin{equation*}
\deg_m(x) =
\sum_{\e\not \in A(F)_x} m(\e) + \sum_{\e\in A(F)_x}m(\e) 
\leq 2d+ C (1+|x|)^{\alpha} \sum_{y\in\Z^{d}}(1+|y|)^{\alpha}. 
\end{equation*}
If $\alpha < -d$, the term $M':=\sum_{y\in\Z^{d}}(1+|y|)^{\alpha}$ is finite and $(1+|x|)^{\alpha}\leq 1$, which implies that $\deg_m$ is bounded.
With the same argument, one obtains that condition \eqref{eq_Cm2} is also satisfied. Indeed, one has
\begin{equation*}
\int_{1}^{\infty} \d\lambda \sup_{\lambda < |x| < 2\lambda}\sum_{\e\in A(F)_{x}}m(\e)
\leq \int_{1}^{\infty} \d\lambda \ CM'(1 + \lambda)^{\alpha}, 
\end{equation*} 
which is also finite if $\alpha < -1$. 
For condition \eqref{eq_Cm3}, observe that  
\begin{align*}
\sum_{\substack{\e\in A(F)_{x} \\ \lambda \leq \lfloor |t(\e)| \rfloor \leq 2\lambda}}m(\e)
& \leq C (1+|x|)^{\alpha} \sum_{\lambda \leq  |y| \leq 2\lambda} (1+|y|)^{\alpha} \\
& \leq C (1+|x|)^{\alpha} \sum_{\lambda \leq  |y| \leq 2\lambda} |y|^{\alpha} \\
&= CM (1+|x|)^{\alpha} \Big( \frac{1}{\alpha + d}(2\lambda)^{\alpha+d} - \frac{1}{\alpha + d}(\frac{\lambda}{\sqrt{d}})^{\alpha+d} + (\frac{\lambda}{\sqrt{d}})^{\alpha+d-1} \Big). 
\end{align*}
Therefore, we obtain 
\begin{equation*}
\int_{1}^{\infty} \d\lambda \sup_{x\in\Z}
\sqrt{\sum_{\substack{\e\in A(F)_{x} \\ \lambda \leq \lfloor |t(\e)| \rfloor \leq 2\lambda}}m(\e)}
\leq \int_{1}^{\infty} \d\lambda \sqrt{CM\Big( \frac{1}{\alpha + d}(2\lambda)^{\alpha+d} - \frac{1}{\alpha + d}(\frac{\lambda}{\sqrt{d}})^{\alpha+d} + (\frac{\lambda}{\sqrt{d}})^{\alpha+d-1} \Big)}.
\end{equation*}
The right hand side is finite if $\frac{\alpha+d}{2} < -1$, that is if $\alpha < -d-2$. 

Let us check the condition \eqref{eq_Cm5} for $s=\frac{3}{4}$. 
Taking $\langle x \rangle^{2s}(1+|x|)^{\alpha} \leq 1$ into account, 
we can show that the condition \eqref{eq_Cm5} is satisfied if $\alpha<-d-\frac{3}{2}$. 
Therefore, under $\alpha<-d-2$, the local wave operators exist and are complete. 
\end{proof}

\begin{proof}[Proof of Example \ref{ex_Toblerone}]
Let us first check that $\deg_m$ is bounded. Indeed, one has for any $x\in V(\X)$
\begin{equation*}
\deg_m(x) =
\sum_{\e\not \in A(F)_x} m(\e) + \sum_{\e\in A(F)_x}m(\e) 
\leq 4+ C \begin{cases} \sum_{y\in V(\X)}(1+|\lfloor y\rfloor|)^{\alpha} & \quad \text{if} \ x=x_{0} \\
(1+|\lfloor x\rfloor|)^{\alpha} & \quad \text{if} \ x\neq x_{0}.   \end{cases}
\end{equation*}
If $\alpha < -1$, $\sum_{y\in V(\X)}(1+|\lfloor y\rfloor|)^{\alpha} = 3\sum_{y\in \Z}(1+|y|)^{\alpha}$ and it is finite, which implies that $\deg_m$ is bounded.
Let us now check that conditions \eqref{eq_Cm2} and \eqref{eq_Cm3} are also satisfied.
Indeed, by considering $|\lfloor x\rfloor|\geq \lambda\geq 1$, we obtain for \eqref{eq_Cm2} 
\begin{equation*}
\int_{1}^{\infty} \d\lambda \sup_{\lambda < |x| < 2\lambda}\sum_{\e\in A(F)_{x}}m(\e)
\leq C \int_{1}^{\infty}  (1+\lambda)^{\alpha} \d\lambda, 
\end{equation*} 
which is finite for $\alpha < -1$. 
For condition \eqref{eq_Cm3}, observe that  
\begin{align*}
\sum_{\substack{\e\in A(F)_{x_{0}} \\ \lambda \leq |\lfloor t(\e) \rfloor| \leq 2\lambda}}m(\e)
& \leq C \sum_{\lambda \leq  |\lfloor y\rfloor| \leq 2\lambda} (1+|\lfloor y\rfloor|)^{\alpha} \\
& \leq 3C \sum_{\lambda \leq  |y| \leq 2\lambda} |y|^{\alpha} \\
&= 3C \Big( \frac{1}{\alpha + 1}(2\lambda)^{\alpha+1} - \frac{1}{\alpha + 1}\lambda^{\alpha+1} + \lambda^{\alpha} \Big). 
\end{align*}
Therefore, we obtain 
\begin{equation*}
\int_{1}^{\infty} \d\lambda \sup_{x\in\Z}
\sqrt{\sum_{\substack{\e\in A(F)_{x} \\ \lambda \leq \lfloor |t(\e)| \rfloor \leq 2\lambda}}m(\e)}
\leq \int_{1}^{\infty} \d\lambda \sqrt{3C \Big( \frac{1}{\alpha + 1}(2\lambda)^{\alpha+1} - \frac{1}{\alpha + 1}\lambda^{\alpha+1} + \lambda^{\alpha} \Big)}.
\end{equation*}
The right hand side is finite if $\frac{\alpha+1}{2} < -1$, that is if $\alpha < -3$. 

Let us check the condition \eqref{eq_Cm5}. 
If $s=\frac{3}{4}$, then we can show that the condition \eqref{eq_Cm5} is satisfied if $\alpha<-\frac{5}{2}$. 
Therefore, under $\alpha<-3$, the local wave operators exist and are complete. 
\end{proof}

\section*{Acknowledgements}

The authors thank D.~Parra for a critical reading of the manuscript, and
for interesting discussions about its content.

\end{document}